\documentclass[aps,pra,superscriptaddress,reprint,nofootinbib,amsmath,amssymb,amsfonts]{revtex4-1}
\usepackage{hyperref,amsthm,mathtools,braket}
\hypersetup{colorlinks=true,linkcolor=blue,citecolor=blue}

\DeclareMathOperator{\Tr}{Tr}
\DeclareMathOperator{\re}{Re}
\DeclareMathOperator{\im}{Im}
\DeclareMathOperator{\erf}{erf}
\DeclareMathOperator{\erfc}{erfc}
\newcommand{\E}{ {\cal E} }
\newcommand{\Mmtm}{\Gamma}
\newcommand{\mmtm}{\gamma}
\newcommand{\abs}[1]{\left| {#1} \right|} 
\newcommand{\ketbra}[2]{\ket{#1}\bra{#2}}
\newcommand{\tr}[2]{\Tr_{#1}\left( #2 \right)}
\DeclareSymbolFont{bbold}{U}{bbold}{m}{n}
\DeclareSymbolFontAlphabet{\mathbbold}{bbold}
\newcommand{\iden}{\mathbbold{1}}
\renewcommand{\doibase}{https://doi.org/}

\theoremstyle{plain}
\newtheorem{thm}{Theorem}
\newtheorem{lem}[thm]{Lemma}

\theoremstyle{definition}

\theoremstyle{remark}

\begin{document}
\date{September 4th, 2019}
\title{Anomalous weak values and contextuality: robustness, tightness, and imaginary parts}
\author{Ravi Kunjwal}
\affiliation{Perimeter Institute for Theoretical Physics, 31 Caroline Street North, Waterloo, ON N2L 2Y5, Canada}
\author{Matteo Lostaglio}
\affiliation{ICFO-Institut de Ciencies Fotoniques, The Barcelona Institute of Science and Technology, Castelldefels (Barcelona), 08860, Spain}
\author{Matthew F. Pusey}
\affiliation{Department of Computer Science, University of Oxford,  Wolfson Building,  Parks Road, Oxford OX1 3QD, UK}

\begin{abstract}
Weak values are quantities accessed through quantum experiments involving weak measurements and post-selection. It has been shown that `anomalous' weak values (those lying beyond the eigenvalue range of the corresponding operator) defy classical explanation in the sense of requiring contextuality [M. F. Pusey, \href{https://doi.org/10.1103/PhysRevLett.113.200401}{Phys. Rev. Lett. \textbf{113}, 200401}, \href{https://arxiv.org/abs/1409.1535}{arXiv:1409.1535}]. Here we elaborate on and extend that result in several directions. Firstly, the original theorem requires certain perfect correlations that can never be realised in any actual experiment. Hence, we provide new theorems  that allow for a noise-robust experimental verification of contextuality from anomalous weak values, and compare with a recent experiment.  Secondly, the original theorem connects the anomaly to contextuality \emph{only} in the presence of a whole set of extra operational constraints. Here we clarify the debate surrounding anomalous weak values by showing that these conditions are tight -- if any one of them is dropped, the anomaly can be reproduced classically. Thirdly, whereas the original result required the \emph{real part} of the weak value to be anomalous, we also give a version for any weak value with nonzero imaginary part. Finally, we show that similar results hold if the weak measurement is performed through qubit pointers, rather than the traditional continuous system.
In summary, we provide inequalities for witnessing nonclassicality using experimentally realistic measurements of \emph{any} anomalous weak value, and clarify what ingredients of the quantum experiment must be missing in any classical model that can reproduce the anomaly.  
\end{abstract}

\maketitle

\section{Introduction}

\emph{Weak measurements} \cite{aharonov88} are a class of minimally disturbing quantum measurements whose practical as well as foundational relevance is currently being investigated~\cite{dressel14}. A weak measurement of an observable $O$ can be realized by weakly coupling a quantum system to a one-dimensional pointer device via a von Neumann-type interaction $\propto O \otimes \Mmtm$, with $\Mmtm$ the momentum of the pointer, so that a small amount of information is imprinted in the pointer at the cost of a small disturbance on the system.

Pivotal to any attempt to establish the presence of nonclassical effects in a given experiment is the formulation of a rigorous no-go theorem based on a precise and operational notion of nonclassicality. It has long been argued that the average final position of the pointer -- conditioned upon a successful postselection performed on the system after the weak measurement -- is a witness to nonclassicality \cite{aharonov88}; in the quantum formalism this quantity is related to the (real part of the) \emph{weak value}, which is ${}_{\phi}\braket{O}_{\psi}:= \braket{\phi|O|\psi}/\braket{\phi|\psi}$, where $O$ is the observable being weakly measured, $\ket{\psi}$ is the initial preparation and $\ket{\phi}$ is the post-selection. A long-standing debate ensued between those supporting the thesis that these experiments are indeed probing truly quantum effects and those arguing that they can be easily understood from classical statistics \cite{leggett89,aharonov89,ferrie14,brodutch15,ferrie15,vaidman17}.  

Recently, a precise no-go theorem was established~\cite{pusey14}. The theorem proves that \emph{anomalous weak values} (AWV), i.e. $ \prescript{}{\phi}{\langle O \rangle}_{\psi} $ taking values beyond the spectrum of $O$, are associated to operational statistics defying any noncontextual explanation in the generalized sense introduced by Spekkens \cite{spekkens05}. Nevertheless, the theorem of Ref.~\cite{pusey14} leaves several questions open:

\begin{enumerate}
	\item First of all, it assumes an \emph{exactly} projective postselection $\ket{\phi}$, which makes any experimental test~\citep{piacentini16} necessarily inconclusive; in fact, any degree of noise makes the no-go theorem inapplicable. Does the nonclassicality of AWV survive real-world conditions? 
	\item Second, both Ref.~\cite{pusey14} and the noise-robust theorems presented here prove that AWV are non-classical in the presence of a set of extra operational conditions. Are these all truly necessary?
	\item Third, the theorem only refers to the \emph{real part} of the weak value. Is a nonzero value of the imaginary part of the weak value also non-classical?
	\item Fourth, the relation between AWV and contextuality holds for a measurement with a continuum of outcomes. Can it be extended to discrete systems, such as an experiment involving only a single qubit pointer, or a coarse graining of the standard weak value experiment? This is also experimentally relevant because the infinitely many operational constraints required for the original theorem \cite{pusey14} to hold cannot be tested by finite means, and a discrete pointer is often more practical anyway.
	\item Finally, the theorem identifies a single noncontextuality inequality which is violated in the presence of AWV. However, is the inequality unique and is it tight?
\end{enumerate}

Our investigation largely answers all these questions:
\begin{enumerate}
  \item We provide two new proofs of contextuality from AWV that are robust to noise, based on Theorems \ref{thm:1} and \ref{th:preparationNC}. The two new proofs are complementary, each requiring the satisfaction of a different set of operational constraints together with the observation of the AWV. These results show that, at the price of extending the set of operational tests required, the relation between AWV and nonclassicality extends beyond the ideal, noiseless case. We also discuss the significance of these results for current experimental tests (Sec.~\ref{currexpts}).
	\item We show that the extra operational conditions in our theorems form a minimal set: dropping any one of them allows to reproduce the AWV within a classical model (Sec.~\ref{sec:necessity}). This illuminates the debate around ``quantumness'' of AWV (e.g. \cite{ferrie14, brodutch15, ferrie15}), since it rigorously shows that it is only in the presence of \emph{all} the operational facts listed in our theorems that AWV defy a classical explanation.
	\item The imaginary part of the weak value admits its own contextuality theorem (Sec.~\ref{imawv}, Theorem \ref{thm:2}). Hence, \emph{any} AWV can be related to contextuality. We clarify why this is not in contradiction with recent studies~\cite{bartlett12, karanjai15} suggesting that imaginary weak values admit a classical model.
	\item The contextuality of AWV has nothing to do with continuous measurements and extends to discrete pointers as well (Sec.~\ref{sec:discrete}, Theorem \ref{thm:3}). This makes the experiment suited for conclusive experimental verification, since in this case only a finite set of operational tests are required.
	\item The noncontextual bound in Ref.~\cite{pusey14} is not tight, but we provide an improved version and investigate its tightness and uniqueness using computational methods from Ref. \cite{schmid17} (Appendix \ref{proof1}).
\end{enumerate}

Our theorems are noise-robust in the sense of not requiring perfectly projective measurements, but noise can also impact the other operational conditions of our proofs. We view those issues as being outside the scope of this work, because with the form of noise-robustness we provide in place, there are generic approaches to tackling the main remaining idealizations, as discussed in Sec.~\ref{remain}.

\section{Noise-robust no-go theorems for anomalous weak values}

\subsection{Weak values}

Let $\rho$ be a quantum state, $O$ an observable and $[y|M_F]$ a post-selection measurement, i.e., $[y=1|M_F] = \Pi_\phi$ (successful post-selection), \mbox{$[y=0|M_F] = \iden - \Pi_\phi$} (failed post-selection), with $\Pi_\phi = \ketbra{\phi}{\phi}$. We can then define the (generalized) \emph{weak value}
\begin{equation}
\label{eq:generalizedweakvalue}
{}_{\phi}{\langle  O \rangle}_{\rho} = \frac{\tr{}{\Pi_\phi O \rho}}{\tr{}{\Pi_\phi\rho}}.
\end{equation}
This expression equals to the standard expression of the weak value of Ref.~\cite{aharonov88} when $\rho = \ketbra{\psi}{\psi}$. For ${}_{\phi}{\langle  O \rangle}_{\rho}$ to be well-defined, we take $\tr{}{\Pi_\phi\rho}>0$, i.e., the preselection and postselection are nonorthogonal. The weak value can be experimentally accessed by a weak measurement of $O$. Specifically, couple $O$ with a one dimensional pointer device through the Hamiltonian $H = O \otimes \Mmtm$, with $\Mmtm$ the momentum of the pointer. Suppose the pointer  is initialized in  a Gaussian pure state centered around the origin and with spread $s$:
	\begin{small}
	\begin{equation}
	\label{eq:gaussianpointer}
	\ket{\psi}_P = \int dx G_{s}(x) \ket{x}, \; \; G_{s}(x)= (\pi s^2)^{-1/4} \exp \left[-x^2/(2s^2)\right].
	\end{equation}
		\end{small}
In the limit $s \rightarrow \infty$, if a projective measurement of the pointer's position is carried out after a unit time, we obtain a so-called \emph{weak measurement} of $O$ ($s \rightarrow 0$ would give a projective measurement of $O$). 

Suppose now the post-selection measurement \mbox{$\{\Pi_\phi, \iden - \Pi_\phi\}$} is carried out on the system, after the interaction with the pointer. The average position of the pointer, conditioned on observing $\Pi_\phi$ (successful post-selection), is proportional to $\re\left({}_{\phi}{\langle O \rangle}_{\rho}\right)$, whereas  $\im\left({}_{\phi}{\langle O \rangle}_{\rho}\right)$ can be recovered from the expected momentum of the pointer given a successful postselection~\cite{jozsa07}.

The weak value is called \emph{anomalous} when it cannot be written as a convex combination of the eigenvalues of $ O$. There are two ways this can happen:
\begin{enumerate}
	\item[(i)]  $\re\left( {}_\phi{\langle O \rangle}_\rho \right)$ is smaller than the smallest eigenvalue of $O$, or larger than the largest eigenvalue,
	\item[(ii)] $\im\left( {}_\phi{\langle O \rangle}_\rho \right) \neq 0$.
\end{enumerate}
Only (i) was related to contextuality in Ref.~\cite{pusey14}, but our results here show that both in fact lead to proofs of contextuality. 

Writing the spectral decomposition of $O$ as $O = \sum_i o_i \E_i$, we have that
\begin{equation}
{}_\phi{\langle O \rangle}_\rho = \sum_i o_i \; {}_\phi{\langle \E_i \rangle}_\rho \label{eq:weakspectral}
\end{equation} and $\sum_i {}_\phi{\langle \E_i \rangle}_\rho = {}_\phi{\langle \iden \rangle}_\rho = 1$. Then, if ${}_\phi{\langle O \rangle}_\rho$ is anomalous, at least one of the ${}_\phi{\langle \E_i \rangle}_\rho$ must be anomalous (i.e. not a standard probability).\footnote{Note that one can have instances in which some or all ${}_\phi{\langle \E_i \rangle}_\rho$ are anomalous, but ${}_\phi{\langle O \rangle}_\rho$ is not, e.g. if an observable has a zero eigenvalue then the weak value of the associated projector is irrelevant to the weak value of the observable.} This is because if all the ${}_\phi{\langle \E_i \rangle}_\rho$ are standard probabilities then \eqref{eq:weakspectral} shows that ${}_\phi{\langle O \rangle}_\rho$ is in the convex hull of the $o_i$.

Since, then, whenever we have an anomalous weak value for an observable $O$ we can also find an anomalous weak value for one of its eigenprojectors, without loss of generality we will focus on weak values of projectors.   

Furthermore, if a projector $\E$ is anomalous due to its real part, then either ${\rm Re} {}_\phi{\langle \E \rangle}_\rho < 0 $ or ${\rm Re} {}_\phi{\langle (\iden - \E) \rangle}_\rho < 0 $; similarly, if a projector $\E$ is anomalous due to its imaginary part, then either ${\rm Im} {}_\phi{\langle \E \rangle}_\rho < 0 $ or 
${\rm Im} {}_\phi{\langle (\iden - \E) \rangle}_\rho < 0 $. 
Hence, without loss of generality we will focus on anomalous weak values for projectors with negative real or imaginary part.

For calculations it will often be useful to refer to the numerator of Eq.~\eqref{eq:generalizedweakvalue}, which we write as $\braket{\Pi_\phi \E}_\rho := \Tr(\Pi_\phi\E\rho)$.\footnote{$\braket{\Pi_\phi\E}_\rho$ coincides with the so-called Kirkwood-Dirac \cite{kirkwood33,dirac45} quasiprobability distribution, the real part of which is the Margenau-Hills \cite{margenau1961correlation} distribution, see Section IV.A of Ref.~\cite{dressel15} for details. These distributions are related to the `optimal' estimate of the observable $\E$ from a measurement of $\Pi_\phi$, under the prior information that the initial state is $\rho$ \cite{hall2004prior}.} Since the denominator $\Tr(\Pi_\phi\rho)$ is a positive real number (in particular, recalling that it must be non-zero for a well-defined weak value), $\braket{\Pi_\phi\E}_\rho$ has negative real or imaginary parts if and only if ${}_\phi{\langle \E \rangle}_\rho$ does.

\subsection{Setting the stage: the standard quantum experiment}\label{stage}

\begin{figure}
	\includegraphics[width=\columnwidth]{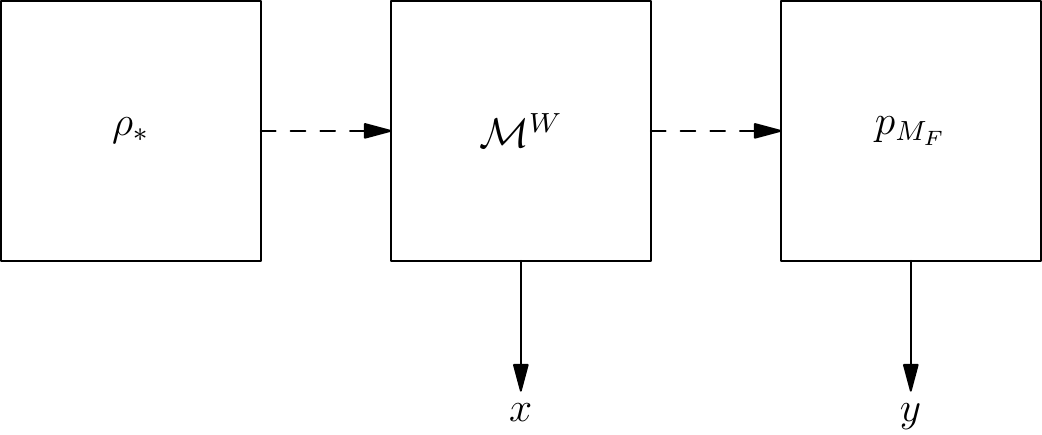}
	\caption{Illustration of the three stages of a quantum weak value experiment.}
	\label{fig:quantum_exp}
\end{figure}

Let us discuss the traditional experimental setting for weak measurements and weak values \cite{aharonov88} (see Appendix~\ref{appendix:standard} for some details of the calculations. Later we will discuss extensions to qubit pointers). As discussed above, we can focus on the weak value of some projector $\E$. There are three stages of the quantum experiment (see Fig.~\ref{fig:quantum_exp}):

\emph{Preparation.} A system is prepared in some quantum state. Since no difficulties arise from allowing a generic mixed state $\rho_*$, we allow mixed preparations. 

\emph{Weak measurement}. A measurement is performed through the following scheme: a pointer device, represented by a one-dimensional continuous system with conjugate variables $X$ and $\Mmtm$, is initialized in the Gaussian pure state $\ket{\psi}_P$ given above. The system is coupled to the pointer through the Hamiltonian $H= \E \otimes \Mmtm$.  

A standard calculation (see, e.g., the proof of Theorem~1 in Ref.~\cite{pusey14}) shows that, after a unit time, a measurement of $X$ on the pointer realises a POVM $[x|M_W] = N^\dag_x N_x$ on the system given by
	\begin{equation}
	\label{eq:N_x}
	N_x = \braket{x|e^{-i H} |\Psi}_P = G_s(x-1) \mathcal{E} +G_s(x)\mathcal{E}^\perp,
	\end{equation}
\begin{equation}
\label{eq:M_Wopeq}
[x|M_W] = G^2_s(x-1)[y=1|M_\E] + G_s^2(x)[y=0|M_\E],
\end{equation}
where $[y=1|M_\E] = \E$, $[y=0|M_\E] = \E^\perp = 1 - \E$.

Let $\mathcal{M}^W_x(\cdot) = N_x (\cdot) N^\dag_x$ be the state update map for outcome $x$. The channel induced by the weak measurement when the outcome is not recorded is 
\begin{equation}
\label{eq:channelM}
\mathcal{M}(\cdot) = \int_{-\infty}^{+\infty} dx \mathcal{M}^W_x (\cdot) =\int_{-\infty}^{+\infty}dx  N_x(\cdot) N^\dag_x.
\end{equation}
One finds, $\mathcal{M}(\rho)
=(1-p_d)\rho+p_d (\E -\mathcal{E}^{\perp})\rho(\mathcal{E}-\mathcal{E}^{\perp})$, with a ``probability of disturbance'' $p_d = \frac{1-e^{-1/4s^2}}{2}$. Hence,
\begin{align}
\label{eq:M_opeq}
\mathcal{M}
=(1-p_d) \mathcal{I} + p_d \mathcal{M}^D,
\end{align}
with $\mathcal{M}^D(\rho):= (\E - \E^\perp)\rho (\E - \E^\perp)$.

\emph{Post-selection}. Finally, one can measure $[y|M_F]$ and compute the probability of a negative $x$ followed by a successful post-selection
\begin{small}
\begin{equation*}
p^{\rm ideal}_- = \int^0_{-\infty} dx \tr{}{ \Pi_\phi N_x \rho_* N^\dag_x} = \int^0_{-\infty} dx \tr{}{ \Pi_\phi \mathcal{M}^W_x(\rho_*)} ,
\end{equation*}
\end{small}
which will be a central witness of nonclassicality in the following theorems.  Denoting the undisturbed probability of post-selection by $p_F=\tr{}{\Pi_\phi\rho_*}$,  one finds 
\begin{equation}
\label{eq:p-ideal}
p^{\rm ideal}_- = \frac{p_F}{2} - \frac{\re\left(\braket{\Pi_\phi\E}_{\rho_*}\right)}{\sqrt{\pi}s} + o\left(\frac{1}{s}\right).
\end{equation}
This is a simple calculation see, e.g., the proof of Lemma~1 in Ref.~\cite{lostaglio18} (note, however, that we redefined $p^{\rm ideal}_-$ without the normalisation by the postselection probability). Recall from the previous section that a weak value with an anomalous real part implies an $\E$ with $\re\left(\braket{\Pi_\phi\E}_{\rho_*}\right) < 0$. We will show that this means $p^{\rm ideal}_-$ is larger than can be explained in a non-contextual model.

\subsection{Non-contextual description of the quantum experiment} 
\label{sec:noncontextualdescription}

We now analyze how a putative non-contextual ontological model (Fig.~\ref{fig:ontological}) would describe the quantum experiment (Fig.~\ref{fig:quantum_exp}). Let us follow the three stages:

\begin{figure}
\includegraphics[width=\columnwidth]{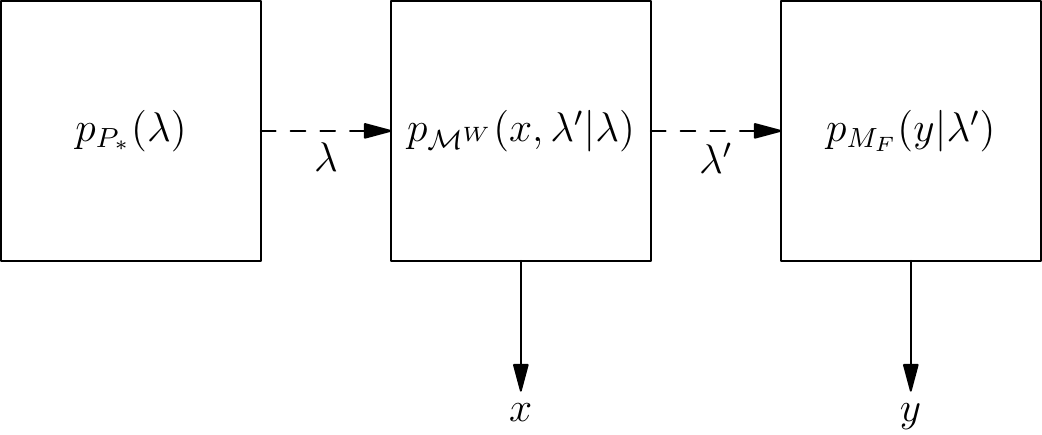}
	\caption{Illustration of an ontological model for the quantum experiment in Fig.~\ref{fig:quantum_exp}.}
	\label{fig:ontological}
\end{figure}

\emph{Preparation.} The preparation of the quantum state $\rho_*$ can be abstractly thought of as a set of instructions $P_*$ that initialize the system. In an ontological model, this is associated to sampling from a distribution $p_{P_*}(\lambda)$ over some set of hidden variables $\lambda$.

\emph{Weak measurement.} The weak measurement is a quantum instrument \{$\mathcal{M}^W_x$\} (also understood as a set of experimental procedures) that, in the ontological model, is represented by the function $p_{\mathcal{M}_W}(x,\lambda'|\lambda)$. This describes the probability that, given as input the state $\lambda$, the weak measurement gives outcome $x$ and updates the state to $\lambda'$ (the update $\lambda \rightarrow \lambda'$ models the potential disturbance induced by the measuring apparatus). If $p_{\mathcal{M}}(\lambda'|\lambda)$ represents the matrix of transition probabilities associated to the channel $\mathcal{M}$ in Eq.~\eqref{eq:channelM}, one has $p_{\mathcal{M}}(\lambda'|\lambda) = \int_{-\infty}^{+\infty} dx p_{\mathcal{M}_W}(x,\lambda'|\lambda)$. On the other hand, the response function $p_{M_W}(x|\lambda)$ of the weak measurement $[x|M_W]$, giving the probability that the weak measurement outputs $x$ given the input state $\lambda$, is given by $p_{M_W}(x|\lambda) = \int d\lambda' p_{\mathcal{M}_W}(x,\lambda'|\lambda)$. 

\emph{Post-selection.} The measurement $[y|M_F]$ is also represented in the ontological model by a response function $p_{\mathcal{M}_F}(y|\lambda)$. While in the quantum experiment $[y|M_F]$ would ideally be a projective measurement, in contrast to Ref.~\cite{pusey14}, our theorems will not rely on this being the case (in fact, our first theorem makes no assumption about $[y|M_F]$). This is necessary in any experimental verification of the relation between anomalous weak values and contextuality, since no experiment can achieve this idealization.

\emph{Operational statistics.} The operational statistics collected by the whole experiment is summarized by
\begin{enumerate}
\item $p(x,y|P_*,\mathcal{M}^W,M_F) $, the probability that if the preparation procedure $P_*$ is followed, sequentially performing the weak measurement procedure $\mathcal{M}^W$ and the post-selection procedure $M_F$, returns outcomes $x$ and $y$, respectively. In the quantum setting this is given by $\tr{}{ [y|M_F] \mathcal{M}^W_x(\rho_*)}$. 
\item $p(y|P_*,M_F) $, the probability that  if the preparation procedure $P_*$ is directly followed by the post-selection measurement procedure $M_F$, one gets outcome $y$. In the quantum setting this is given by $\tr{}{ [y|M_F] \rho_*}$. 
\end{enumerate}
An ontological model for this experiment is a set of assignments as described above and satisfying
\begin{multline*}
p(x,y|P_*,\mathcal{M}^W,M_F) =\\ \int d\lambda' d\lambda p_{P_*}(\lambda)  p_{\mathcal{M}_W}(x,\lambda'|\lambda) p_{M_F}(y|\lambda'),
\end{multline*}
\begin{equation*}
p(y|P_*,M_F) = \int d\lambda p_{P_*}(\lambda) p_{M_F}(y|\lambda).
\end{equation*}

\emph{Noncontextuality.} A generic ontological model description of the experiment can always be found, whatever the operational statistics. However, non-contextual models (according to the definition of Ref.~\cite{spekkens05}) are those that \emph{associate to operationally indistinguishable procedures identical representation in the ontological model}. In the present case, the weak measurement procedure $[x|M_W]$ is operationally equivalent, due to Eq.~\eqref{eq:M_Wopeq}, to measuring $[y|M_\E]$ and then sampling as prescribed according to the distribution $G^2_s(x)$. Hence non-contextual models require
\begin{equation}
\label{eq:noncontextualMw}
p_{M_W}(x|\lambda) = G^2_s(x-1) p_{M_\E}(y=1|\lambda) + G^2_s(x) p_{M_\E}(y=0|\lambda),
\end{equation}
where $p_{M_\E}(y|\lambda)$ is the response function of the measurement $[y|M_\E]$. Similarly, the operational equivalence of Eq.~\eqref{eq:M_opeq} implies that non-contextual models satisfy
\begin{equation}
\label{eq:noncontextualM}
p_{\mathcal{M}}(\lambda'|\lambda) = (1-p_d) p_{\mathcal{I}}(\lambda'|\lambda) + p_d p_{\mathcal{M}^D}(\lambda'|\lambda),
\end{equation}
where $p_{\mathcal{I}}(\lambda'|\lambda)$ and $ p_{\mathcal{M}^D}(\lambda'|\lambda)$ are matrices of transition probabilities representing the channels $\mathcal{I}$ and $\mathcal{M_D}$ in the ontological model.

\subsection{AWV and contextuality beyond idealisations}

In this section we will start our investigation by presenting two results. First, the assumption of non-contextuality limits the maximum value achievable by the quantity  
\begin{equation}
p_-:=  \int_{-\infty}^0 p(x,y=1| P_*, \mathcal{M}^W,M_F) dx,
\end{equation}
even beyond the idealized scenario studied in Ref.~\cite{pusey14}. Secondly, in the quantum treatment the relation between $p_-$ and the weak value presented in Eq.~\eqref{eq:p-ideal} extends to situations where noise and imperfections are present. Combining these two results we obtain our first proof that (real) anomalous weak values are nonclassical beyond the idealized setting of Ref.~\cite{pusey14}. What is more, we can quantify how strong the anomaly needs to be, for given noise, to prove contextuality.

To highlight the independence of our noncontextuality theorems from the quantum formalism, we introduce the notation $\simeq$ to denote operationally indistinguishable procedures, following Ref.~\cite{spekkens05}. For example, instead of the operator equality of Eq.~\eqref{eq:M_Wopeq} we will write
\begin{equation*}
[x|M_W] \simeq G^2_s(x-1)[y=1|M_\E] + G_s^2(x)[y=0|M_\E],
\end{equation*}
which means that the above two measurement procedures give rise to the same operational statistics for every preparation procedure taken as input. Similarly, Eq.~\eqref{eq:M_opeq} becomes 
\begin{equation*}
\mathcal{M}
\simeq (1-p_d) \mathcal{I} + p_d \mathcal{M}^D,
\end{equation*}
denoting that, for any preparation procedure used to initialize the system, if we apply either of the above two transformations and then measure according to an arbitrary measurement procedure, the outcome statistics will be identical.
When the relevant operational data arises from quantum experiments, however, $\simeq$ can be simply identified with the corresponding operator identities, as we did in the previous section.

\begin{thm}[Noise-robust contextuality from the real part of the weak value]
	\label{thm:1}
	Suppose we have a noncontextual ontological model and that:
	\begin{enumerate}
	  \item \label{condition:thm1} There exists a 2-outcome measurement $M_\E$ and a probability distribution $q(x)$ with median $x=0$ such that, for all $x \in \mathbb{R}$,
	     \begin{equation}
		    [x|M_W] \simeq q(x-1)[y=1|M_{\E}] + q(x)[y=0|M_{\E}].
		    \label{eq:thm1}
	      \end{equation}
		
		\item \label{condition:thm1_2} If $\mathcal{M}:= \int \mathcal{M}^W_x dx$, there exists $p_d \in [0,1]$ such that
		\begin{equation}
			\mathcal{M} \simeq (1-p_d) \mathcal{I} + p_d \mathcal{M}^D, 
			\label{eq:thm1_2}
		\end{equation}
		where $\mathcal{I}$ denotes the identity transformation and $\mathcal{M}^D$ some other transformation.
	\end{enumerate}
	Then, if $p_-:=  \int_{-\infty}^0 p(x,y=1| P_*, \mathcal{M}^W,M_F) dx$ and $p_F := p(y=1|P_*,M_F)$,
	\begin{equation}
		p_- \leq p^{NC}_- := p_F\frac12  + (1-p_F)p_d .
	\end{equation}
\end{thm}
It follows from the first assumption that the marginal probability of the weak measurement $M_W$ giving a negative result is at most $\frac12$. If the system was totally undisturbed then the post-selection would occur independently with probability $p_F$. This would give a joint probability of negative result and post-selection of at most $\frac{P_F}2$. Our inequality shows that noncontextual models cannot explain measurement-disturbance increasing the joint probability $p_-$ above this no-disturbance bound by more than $O(p_d)$. 

We provide the proof of this theorem in Appendix~\ref{proof1}. However, to give some intuition we give here a simplified proof that holds for a finite ontic state space and only derives a weaker noncontextuality bound (but still strong enough that it suffices to prove that real anomalous weak values are contextual):
\begin{proof}[Proof (simplified version)]
		In the ontological model
		\begin{equation}\label{finiteform}
		p_- = \int_{-\infty}^0 \sum_{\lambda',\lambda} p_{M_F}(y=1|\lambda')p_{\mathcal{M}^W}(x, \lambda'|\lambda)p_{P_*}(\lambda)dx, 
	    \end{equation}
		and
		\[
		  p_F = \sum_\lambda p_{M_F}(y=1|\lambda) p_{P*}(\lambda).
		\]
As discussed in Sec.~\ref{sec:noncontextualdescription},
		 $p_{M_W}(x|\lambda) = \sum_{\lambda'}p_{\mathcal{M}^W}(x, \lambda'|\lambda)$. Hence $p_{M_W}(x|\lambda) \geq p_{\mathcal{M}^W}(x, \lambda|\lambda)$. Using measurement noncontextuality and Eq.~\eqref{eq:thm1}, one obtains Eq.~\eqref{eq:noncontextualMw}, i.e.
\[
p_{M_W}(x|\lambda) = G^2_s(x-1) p_{M_\E}(y=1|\lambda) + G^2_s(x) p_{M_\E}(y=0|\lambda).
\]
		 Since $G^2_s(x)$ has median zero, this immediately implies 
		 \begin{equation*}
		   \int_{-\infty}^0p_{M_W}(x|\lambda) dx \leq \frac{p_{M_\E}(y=1|\lambda) + p_{M_\E}(y=0|\lambda)}2 = \frac12.
		 \end{equation*} 
		 Hence, for the terms in Eq.~\eqref{finiteform} with $\lambda' = \lambda$ we have
		\begin{align*}
	&	\int_{-\infty}^0 \sum_{\lambda} p_{M_F}(y=1|\lambda)p_{\mathcal{M}^W}(x, \lambda|\lambda)p_{P_*}(\lambda)dx  \\ &\leq
		\int_{-\infty}^0 \sum_{\lambda} p_{M_F}(y=1|\lambda)	p_{M_W}(x|\lambda) p_{P_*}(\lambda)dx \\
		&\leq  \frac{1}{2}  \sum_{\lambda} p_{M_F}(y=1|\lambda)	p_{P_*}(\lambda)  = \frac{p_F}{2}.
		\end{align*}
		Furthermore, as discussed in Sec.~\ref{sec:noncontextualdescription}, $p_\mathcal{M}(\lambda'|\lambda) = \int_{-\infty}^\infty p_{\mathcal{M}^W}(x, \lambda'|\lambda)$, hence $p_\mathcal{M}(\lambda'|\lambda)  \geq \int_{-\infty}^0 p_{\mathcal{M}^W}(x, \lambda'|\lambda)dx$. By Eq.~\eqref{eq:thm1_2} and transformation noncontextuality we have Eq.~\eqref{eq:noncontextualM}, i.e.
		\[
p_{\mathcal{M}}(\lambda'|\lambda) = (1-p_d) p_{\mathcal{I}}(\lambda'|\lambda) + p_d p_{\mathcal{M}^D}(\lambda'|\lambda).
\]
Then, since $p_{\mathcal{I}}(\lambda'|\lambda) = \delta_{\lambda'\lambda}$ (e.g. using noncontextuality and taking into account that $\mathcal{I}$ can be implemented by letting no time pass, so that no dynamics can occur) one has that, for $\lambda' \neq \lambda$, 	$p_{\mathcal{M}}(\lambda'|\lambda) =  p_d p_{\mathcal{M}^D}(\lambda'|\lambda)$. Hence, for the terms of Eq.~\eqref{finiteform} with $\lambda' \neq \lambda$ we have
			\begin{align*}
			  &	\int_{-\infty}^0 \sum_{\lambda} \sum_{\lambda' \neq \lambda} p_{M_F}(y=1|\lambda')p_{\mathcal{M}^W}(x, \lambda'|\lambda)p_{P_*}(\lambda)dx  \\
			  & \leq  \sum_{\lambda} \sum_{\lambda' \neq \lambda} p_{M_F}(y=1|\lambda')p_{\mathcal{M}}( \lambda'|\lambda)p_{P_*}(\lambda) \\
			  &  = p_d  \sum_{\lambda} \sum_{\lambda' \neq \lambda} p_{M_F}(y=1|\lambda')p_{\mathcal{M}^D}(\lambda'|\lambda)p_{P_*}(\lambda) \\
			  & \leq p_d \sum_{\lambda} \sum_{\lambda' \neq \lambda} p_{\mathcal{M}^D}(\lambda'|\lambda)p_{P_*}(\lambda) \\
			  & \leq p_d \sum_{\lambda} p_{P_*}(\lambda) \\
			  & = p_d.
		\end{align*}
		Summing the $\lambda' = \lambda$ and $\lambda' \neq \lambda$ terms gives \mbox{$p_- \leq p_F/2 + p_d$}.
	\end{proof}

Our first illustration of how this theorem operates is in the idealized scenario discussed above. First, the operational equivalences in Eq.~\eqref{eq:thm1} and Eq.~\eqref{eq:thm1_2} are satisfied with $q(x) = G^2_s(x)$, due to Eq.~\eqref{eq:M_Wopeq} and Eq.~\eqref{eq:M_opeq}, respectively. Furthermore, $p_d = \frac{1-e^{-1/4s^2}}{2} = o(1/s^2)$. Hence, from the above theorem, the data can only be explained by a non-contextual ontological model if the probability $p_-$ of passing the postselection and displaying a negative pointer position is 
\begin{equation}
p_- \leq p_F /2 + o(1/s).
\end{equation} 

However, quantum mechanically $p_- = 
p^{\rm ideal}_-$ as given by Eq.~\eqref{eq:p-ideal}. When $\re \left(\braket{\Pi_\phi\E}_{\rho_*}\right)\geq 0$, $p_-$ is always smaller than $p_F/2$ for $s$ large enough. However, whenever $\re \left(\braket{\Pi_\phi\E}_{\rho_*}\right)<0$ (anomalous real weak value) there exists an $s$ large enough for which $p^{\rm ideal}_- > p_F /2 + o(1/s)$, from which we obtain a proof of contextuality.  


Note already that this statement does not require the preparation to be pure, as is the case in standard formulations. However, going beyond this, our theorem does not require the post-selection to be exactly projective either. 
For example, let us assume that unbiased noise is present in the post-selection, i.e. in the quantum description,
\begin{multline}
\label{eq:noisypostselection-main}
\{[y=1|M_F],[y=0|M_F] \} =\\ (1-2\epsilon)\{\Pi_\phi, \iden - \Pi_\phi\} + 2\epsilon \{\iden/2,\iden/2\},
\end{multline} 
where $\epsilon\in\left(0,\frac{1}{2}\right)$.
We show in Appendix~\ref{quantval} that the operational equivalences of Eq.~\eqref{eq:thm1} and \eqref{eq:thm1_2},  are still satisfied and, furthermore, 
\begin{equation}
 p_- = p^{\rm noisy}_- := \frac{p_F}{2}  - \frac{1}{\sqrt{\pi}s}\re\left(\braket{[y=1|M_F] \E}_{\rho_*}\right)   + o\left(\frac{1}{s}\right) .\label{eq:quantump-}
\end{equation}

Hence, if $p^{\rm noisy}_- > p^{NC}_- = p_F/2 + o(1/s)$ the experiment still provides a proof of contextuality. As is intuitive, $p^{\rm noisy}_- $ is determined by a \emph{noisy weak value}, whose relation with the ideal one can be inferred from
\begin{equation}
\nonumber
\re\left(\braket{[y=1|M_F] \E}_{\rho_*}   \right) = (1-2\epsilon) \re\left(\braket{\Pi_\phi\E}_{\rho_*}\right) + \epsilon p_\E, 
\end{equation}
where $p_\E :=\tr{}{\E \rho_*}$. This clarifies that the noise, parametrized by $\epsilon$, linearly `damps' the potential negativity of the weak value. In fact, using $\re\left(\braket{\Pi_\phi\E}_{\rho_*}\right) \geq -1/8$ (Eq.~(41) of Ref.~\cite{allahverdyan2014nonequilibrium}), we can estimate the noise threshold for $p^{\rm noisy}_- > p^{NC}_-$ in Theorem~\ref{thm:1} to be $\epsilon < \frac{1}{2 + 8 p_\E}$.   

As an experimental proposal, one can consider the setup of Ref.~\cite{piacentini16}. The measured $p^{\rm noisy}_-$ is well above $p^{\rm NC}_-$. Hence, if Eqs.~\eqref{eq:thm1} and \eqref{eq:thm1_2} were verified (only Eq.~\eqref{eq:thm1} is claimed), the experiment would be a proof of contextuality from AWV. The importance of checking all the operational equivalences $\simeq$ will be stressed later (Sec.~\ref{sec:necessity}), when we show that, if even one of them is dropped, a classical model exists reproducing the anomaly. 

We conclude this section by discussing in more detail the relation between Theorem~\ref{thm:1} and the main theorem of Ref.~\cite{pusey14}.  One can note that Eq.~\eqref{eq:thm1} is exactly the first operational equivalence used in Ref.~\cite{pusey14}, while Eq.~\eqref{eq:thm1_2} is a stronger operational requirement than the second equivalence of Ref.~\cite{pusey14}, as Eq.~\eqref{eq:thm1_2} involves the transformation rather than the measurement. Importantly, Theorem~\ref{thm:1} makes no reference to the properties of $[y|M_F]$ (for example, we do not require any of the properties associated with projective measurements in quantum theory), which is what allowed the above discussion of experimental proofs of contextuality from AWV in non ideal scenarios. As a minor difference, the inequality derived in Ref.~\cite{pusey14} is\footnote{Notice that what was called $p_-$ in \cite{pusey14} is what we call $\frac{p_-}{p_F}$ here.} $p_- \leq \frac12 p_F + p_d$, whereas we now obtain $p_- \leq \frac12p_F + (1-p_F)p_d$. Since $0 < 1-p_F < 1$ the new bound is strictly stronger, although because $p_d$ will typically be very small the improvement is minor. We will later provide evidence that the improved bound is tight.

\subsection{Contextuality from imaginary weak value}\label{imawv}

Our second theorem concerns the imaginary part of the weak value. The theorem of Ref.~\cite{pusey14} does not imply any connection between $\im\left( {}_\phi{\langle O \rangle}_\rho \right) \neq 0$ and contextuality; furthermore, the imaginary weak value has analogues in classical models \cite{bartlett12, karanjai15}. Nevertheless, we show that quantum mechanical imaginary weak values are contextual. This complements the results of Ref.~\cite{pusey14} by showing that \emph{every} anomalous weak value is nonclassical -- not just those with an anomalous real part. 

Let us recall how the imaginary part of a weak value is accessed experimentally. Suppose we keep the same initial pointer state $\ket{\Psi}_P$ as Eq.~\eqref{eq:gaussianpointer} and same interaction Hamiltonian $H = \E \otimes \Mmtm$ applied for a unit time ($\E$ is some projector and $\Mmtm$ the momentum operator of the pointer). However, the pointer is measured in the momentum basis $\{\ket{\mmtm}\}$. This gives a POVM $[\mmtm|M_W] = N^\dag_\mmtm N_\mmtm$ on the system with:
\begin{equation}
  N_\mmtm =\braket{\mmtm|e^{-i H} |\Psi}_P= \braket{\mmtm|\Psi}_P(\exp(-i\mmtm)\E + \E^\perp),
  \label{eq:N_p-mainkraus}
\end{equation}
so that
\begin{equation}
  [\mmtm|M_W] = \abs{\braket{\mmtm|\Psi}_P}^2 (\E + \E^\perp) = \abs{\braket{\mmtm|\Psi}_P}^2 \iden,
  \label{eq:N_p-main}
\end{equation}
with $\braket{\mmtm|\Psi}_P = \pi^{-1/4}\sqrt{s}\exp\left(-\frac{p^2s^2}2\right)$. Note that these are exactly the POVM elements for a trivial measurement sampling from the probability distribution $\abs{\braket{\mmtm|\Psi}_P}^2$, which has median zero.

The choice of measurement on the pointer does not affect the marginal channel on the system and so Eq.~\eqref{eq:M_opeq} is still satisfied with $p_d = o(1/s)$:
\begin{multline}
\label{eq:M-main}
\mathcal{M}(\cdot) = \int_{-\infty}^{+\infty} d\mmtm \mathcal{M}^W_\mmtm(\cdot) = \int_{-\infty}^{+\infty} d\mmtm N_\mmtm(\cdot) N^\dag_\mmtm 
\\=(1-p_d) \mathcal{I}(\cdot) + p_d \mathcal{M}^D(\cdot).
\end{multline}
Furthermore, we show in Appendix~\ref{imaginaryappendix} that in the ideal case we obtain a negative momentum and successful postselection with probability
\begin{equation}
\label{eq:p-momentum}
p^{\rm ideal}_- = \frac{p_F}{2} - \frac1{\sqrt{\pi}s}\im\left(\braket{\Pi_\phi \E}_{\rho_*}\right) + o\left(\frac{1}{s}\right).  
\end{equation}

Given the above setting, that nonzero imaginary values of the weak values are a proof of contextuality is a consequence of the following theorem, proven in Appendix~\ref{proof1}:
\begin{thm}
  \label{thm:2}
	Suppose we have a noncontextual ontological model and that:
	\begin{enumerate}
	   \item If $M_\text{triv}$ involves ignoring the system and sampling a $\mmtm \in \mathbb{R}$ that is negative with probability $\frac12$,
	   \begin{equation}
	  [\mmtm|M_W] \simeq [\mmtm|M_\text{triv}].
	  \label{eq:thm2}
	  \end{equation}
		\item If $\mathcal{M}:= \int \mathcal{M}^W_\mmtm d\mmtm$, there exists $p_d \in [0,1]$ such that
		\begin{equation}
			\mathcal{M} \simeq (1-p_d) \mathcal{I} + p_d \mathcal{M}^D, 
			\label{eq:thm2_2}
		\end{equation}
		where $\mathcal{I}$ denotes the identity transformation and $\mathcal{M}^D$ some other transformation.
	\end{enumerate}
	Then if $p_-:=  \int_{-\infty}^0 p(\mmtm,y=1| P_*, \mathcal{M}^W,M_F) d\mmtm$ and $p_F := p(y=1|P_*,M_F)$,
	\begin{equation}
	\label{eq:NCinequalityimaginary}
		p_- \leq p^{NC}_- = p_F\frac12  + (1-p_F)p_d.
	\end{equation}
\end{thm}

Note that the result requires no mention of $M_\E$. The assumptions of the theorem are satisfied by the experimental setting measuring the imaginary part of the weak value. In fact, Eq.~\eqref{eq:thm2} and Eq.~\eqref{eq:thm2_2} follow immediately from Eq.~\eqref{eq:N_p-main} and Eq.~\eqref{eq:M-main}, respectively. Hence, since $p_d=o(1/s)$, as before if we observe $p_- > p^{NC}_- = p_F\frac12 + o(1/s)$ we have a proof of contextuality. From Eq.~\eqref{eq:p-momentum} this happens whenever $\im\left({}_{\phi}{\langle \E \rangle}_{\rho}\right)$ is negative (recall from Sec.~\ref{stage} that an non-zero imaginary part can be taken negative without loss of generality). Hence, imaginary weak values are contextual. Together with the theorem of the previous section, this shows that \emph{all} (real or imaginary) anomalous weak values are contextual. The theorem also covers noisy post-selection in exactly the same way as we discuss after Theorem~\ref{thm:1}, with $p^{\rm ideal}_-$ of Eq.~\eqref{eq:p-momentum} substituted by a noisy analogue involving a `noisy imaginary weak value', as discussed in Eq.~\eqref{eq:quantump-} for the real part.

%
%

\subsubsection*{The status of imaginary weak values}
 
This result contrasts with the dismissal of the imaginary parts of weak values in Ref.~\cite{pusey14}. The discussion there begins by pointing out that 
\begin{quote}
	``the imaginary part [\ldots] is manifested very differently from the real part \cite{jozsa07}.''
\end{quote}

This is true, and explains why the proof of contextuality has to be adapted slightly to apply to this case. More formally, we could note that the relevant Kraus operators of the weak measurement on the system when we access the pointer's momentum 
are proportional to unitaries $\exp(-i\E \mmtm)$ (see Eq.~\eqref{eq:N_p-mainkraus}). Hence, the same instrument could be achieved by classically sampling an ``outcome'' $\mmtm$ (as in $M_\text{triv}$ above) and then directly performing the appropriate unitary. When we do things this way, it is clear that the correlation between the sampled outcome and the post-selection is purely due to the fact we have disturbed the system by applying a unitary. Yet, since the same instrument is implemented as in the measurement of the imaginary part of the weak value, the same proof of contextuality holds for this sampling scheme. Nonclassicality arises in this case because the unitaries are strong enough to significantly affect the post-selection and yet they average out to something very close to the identity channel. Whilst the leading-order effect of the unitaries on the post-selection is captured by exactly the imaginary part of the weak value, if one has actually implemented the instrument by applying various unitaries it is unclear why this should be expected to reveal anything about the ``value'' of the system observable.

This brings us to the next sentence of Ref.~\cite{pusey14}, which gives a specific argument against imaginary weak values being nonclassical:
\begin{quote}
	``Indeed complex weak values are easily obtained even in the Gaussian subset of quantum mechanics, which has weak measurements (with the same information-tradeoff disturbance [sic] utilised here) and yet admits a very natural non-contextual model \cite{bartlett12}.''
\end{quote}
Weak measurements in the referenced model have since been explored in detail by Karanjai \emph{et.\@ al.\@} \cite{karanjai15}. The model gives definite values to all the allowed observables and so one can meaningfully talk about what values observables truly have independently of any measurement. It is found in Ref.~\cite{karanjai15} that the real part of the weak value reflects the true average value of the observable given the information from the preparation and postselection. Imaginary parts can also arise, but they are purely an artefact of disturbance, in agreement with the discussion above.

Since the model in Ref.~\cite{bartlett12} is noncontextual, the Gaussian subset of quantum mechanics cannot violate any noncontextuality inequalities. But the weak values in the theory do have imaginary parts, and the weak measurements thereof satisfy Eq.~\eqref{eq:thm2}. Therefore the measurements must fail to satisfy Eq.~\eqref{eq:thm2_2} with a sufficiently small $p_d$. In other words, if we measure disturbance using $p_d$ then, contrary to the claim in parenthesis in the quotation above, the weak measurements considered in Ref.~\cite{karanjai15} do \emph{not} have the favorable information-disturbance tradeoff needed to prove contextuality.

We should clarify that this is not in contradiction with our calculations of $p_d$ because those calculations are only valid for the weak measurement of a projector, which has eigenvalues $0$ and $1$. The only observables that can be weakly measured in Ref.~\cite{bartlett12} are linear combinations of position and momentum operators, which all have unbounded spectrum. To get some intuition for why this makes a difference, we can easily generalize the calculation of $\mathcal{M}$ to the case of measuring an operator $O = \sum_i o_i\E_i$ with an arbitrary finite number of eigenvalues $\{o_i\}$, giving
\begin{equation}
\mathcal{M}(\rho) = \sum_{i,j} \exp\left( -\frac{(o_i - o_j)^2}{4s^2} \right)\E_i \rho \E_j.
\end{equation}
Since $\mathcal{I}(\rho) = \sum_{i,j} \E_i \rho \E_j$, to satisfy Eq.~\eqref{eq:thm2_2} we must have $\mathcal{M^D}(\rho) = \sum_{i,j} C_{ij} \E_i \rho \E_j$ with
\begin{equation}
C_{ij} = \frac{1}{p_d}\left(\exp\left( -\frac{(o_i - o_j)^2}{4s^2} \right) - (1 - p_d)\right).\label{eq:Cdef}
\end{equation}
Notice that the Choi-Jamiolkowski state associated to $\mathcal{M}^D$ has a block-diagonal structure in which $C_{ij}$ appear. Hence, $\mathcal{M}^D$ is completely positive if and only if $C_{ij}$
are the entries of a positive matrix. In particular this requires $\abs{C_{ij}} \leq \frac{C_{ii} + C_{jj}}2 = 1$, where $C_{ii} = 1$ follows directly from Eq.~\eqref{eq:Cdef}. The requirements that $C_{ij} \geq -1$ for all $(i,j)$ can be written
\begin{equation}
p_d \geq \frac12 \left( 1 - \min_{i,j}\exp\left( -\frac{(o_i - o_j)^2}{4s^2} \right) \right).
\end{equation}
Hence as we increase the difference between the smallest and largest $o_i$, we need a larger $s$ to ensure a small $p_d$. This suggests that for operators with an unbounded spectrum we should expect that Eq.~\eqref{eq:thm2_2} can only be satisfied with $p_d \geq \frac12$, which is far too large to allow a violation of the noncontextuality inequality in Eq.~\eqref{eq:NCinequalityimaginary}.

\subsection{AWV and contextuality with qubit pointers or coarse graining}\label{sec:discrete}  

While Theorem~\ref{thm:1} removed the idealizations of a perfectly projective postselection and pure input states from the main result of Ref.~\cite{pusey14}, we still followed the traditional approach of introducing weak values using a continuous variable pointer, see Sec.~\ref{stage}. Correspondingly, Theorem~\ref{thm:1} strictly requires an infinite number of operational equivalences to be satisfied, which cannot be checked by finite means. In the following, we will solve this issue. 

It is known that one can follow an experimental setting for measuring weak values that is analogous to the one discussed above but uses a qubit pointer only \cite{wu09}; alternatively, one can consider a coarse graining of $x$ in the traditional setting of Sec.~\ref{stage}. Either way, in these alternative scenarios with finite degrees of freedom we are able to prove that (1) the connection between AWV and contextuality holds and (2) as opposed to Theorem~\ref{thm:1}, the no-go theorem only requires to verify a finite number of operational equivalences. The relevant no-go theorem, proven in Appendix~\ref{proof1}, is given by the following:

\begin{thm}[Noise-robust no-go theorem -- finite version]
	\label{thm:3}
	Suppose we have a noncontextual ontological model and that:
	\begin{enumerate}
	  \item There exists a measurement $M_\E$ and a probability $p_m$ such that
	    \begin{equation}
	    [x|M_W] \simeq p_m [x|M_{\E}] + (1-p_m)[x|M_\text{triv}].
	    \label{eq:thm3}
	    \end{equation}
	    where $M_\text{triv}$ involves ignoring the system and sampling an $x$ that is negative with probability $\frac12$.
	  \item If $\mathcal{M}:= \int \mathcal{M}^W_x dx$, there exists $p_d \in [0,1]$ such that
	    \begin{equation}
		    \mathcal{M} \simeq (1-p_d) \mathcal{I} + p_d \mathcal{M}^D, 
		    \label{eq:thm3_2}
	    \end{equation}
	    where $\mathcal{I}$ denotes the identity transformation and $\mathcal{M}^D$ some other transformation.
	\end{enumerate}
	Then if $p_-:=  \int_{-\infty}^0 p(x,y=1| P_*, \mathcal{M}^W,M_F) dx$ and $p_F := p(y=1|P_*,M_F)$,
	\begin{equation}
	  p_- \leq p_F\frac{1 + p_m}2  + (1-p_F)p_d.
	\end{equation}
\end{thm}

In appendix~\ref{qubitappendix} we describe a weak measurement scheme using a qubit pointer with small parameter $\epsilon$. The outcome is a discrete $x = \pm 1$ so the integrals over $x$ above reduce to sums. We show that the operational equivalences of Eq.~\eqref{eq:thm3} and Eq.~\eqref{eq:thm3_2} are satisfied with $p_m = 2\epsilon + o(\epsilon)$ and $p_d = o(\epsilon)$ respectively, and calculate
\begin{equation}
  p_- = p_F \frac{1+p_m}{2} - 2\epsilon\re\left(\braket{\Pi_\phi\E}_\rho\right) + o(\epsilon),
\end{equation}
giving contextuality for sufficiently small $\epsilon$ whenever $\re\left(\braket{\Pi_\phi\E}_\rho\right) < 0$, as before.

The same argument can be made for the standard quantum experiment described in Sec.~\ref{stage}, once we coarse grain the pointer position to a two outcome measurement $M^{\rm coarse}_W$ with outcomes $x \leq 1/2$ and $x \geq 1/2$ (i.e., $x$ closest to the eigenvalue $0$ of $\E$, or closest to the eigenvalue $1$). If we now label these outcomes $x = -1$ and $x=+1$ respectively, then the conditions of Theorem~\ref{thm:3} are satisfied with $p_d = o(1/s)$ and $p_m = 1/(\sqrt{\pi}s) + o(1/s)$. Then, for the perfect postselection,
\begin{equation}
  p_- = p_F\frac{1+p_m}2 - \frac{1}{\sqrt{\pi}s}\re\left(\braket{\Pi_\phi \E}_{\rho_*}\right) + o\left( \frac1s \right),   
\end{equation}
which, with large $s$, violates the noncontextuality bound.

\subsection{A remark on the debate concerning AWV}

Theorem~\ref{thm:3} not only tells us that weak value experiments proving contextuality can be conducted with qubit pointers, but also clarifies another issue of the weak value debate. When  Ferrie and Combes presented \emph{discrete} classical toy models reproducing certain aspects of AWV~\cite{ferrie14}, questions were posed if these are good analog of the weak value due the intrinsic discreteness \cite{brodutch15} (as opposed to the standard quantum experiment which is continuous or, when discrete, it is a coarse graining of a continuous measurement~\cite{lu2014experimental}). Theorem~\ref{thm:3} shows that the contextuality of the weak value has nothing to do with the fact that we are performing a measurement of a continuous quantity -- the pointer position or momentum: nonclassicality is present both in the coarse-graining of the standard experiment as well as in an intrinsically discrete experiment. In particular, although the weak value no longer appears simply as an average pointer position, the correct ``scaling procedure'' to determine whether a discrete outcome is sufficiently biased to be considered anomalous can be determined operationally using $p_m$.

\section{An alternative approach to the no-go theorems}
\subsection{Theorem based on measurement and preparation noncontextuality}

In Theorems~\ref{thm:1}-\ref{thm:3} we removed the idealisation of exact post-selection in Ref.~\cite{pusey14} and extended an operational equivalence on a measurement to a correspondent operational equivalence on a transformation, Eq.~\eqref{eq:thm1_2}. In fact, Eq.~\eqref{eq:thm1_2} requires us to check that \emph{every} subsequent measurement on the system is affected little by the weak measurement, whereas the original assumption only required to check that the post-selection is affected little when preceded by the weak measurement. Here we present an alternative approach in which we keep the original, less demanding, assumptions of Ref.~\cite{pusey14}, but we introduce some extra preparations whose aim is to provide an operational measure of how `close to projective' the post-selection is.\footnote{This general strategy to `robustify' contextuality proofs was first proposed in Ref.~\cite{kunjwal15}.}

To do so, we 
\begin{enumerate}
	\item Introduce an ensemble of preparations $[b|S]$, where \mbox{$[b=0|S]$} is prepared with probability $q_0$ and \mbox{$[b=1|S]$} is prepared with probability $q_1=1-q_0$. In practice, we will look for $S$ that maximises the correlations with the corresponding outcomes of the (imperfect) post-selection, i.e. maximising
		\begin{equation*}
		C_{S} := p(b=0,y=0|S,M_F) + p(b=1,y=1|S,M_F),
		\end{equation*}
	where $p(b,y|S,M_F)$ is the probability that $[b|S]$ is prepared and an immediate measurement of $[y|M_F]$ on $[b|S]$ returns outcome $y$.
\item If $P_*$ denotes the input preparation in the standard setting (as  in Sec.~\ref{sec:noncontextualdescription}), include it into an ensemble where $P_*$ is prepared with probability $q_*$  and $P_\perp$ is prepared with probability $q_\perp = 1- q_*$. $P_\perp$ and $q_*$ are chosen such that $q_0[b=0|S] + q_1[b=1|S] \simeq  q_*P_* + q_\perp P_\perp$.
\end{enumerate}

It is useful to spell out what this means in quantum terms when the system being weakly measured is a qubit. We start with $\{M_F,\iden - M_F \}$, the imperfect post-selection POVM, and the preparation $\rho_*$. Then we look for states $\sigma_b$, $b=0,1$, that maximize $\Tr(M_F\sigma_1)$ and $\Tr((\iden-M_F)\sigma_0)$.
	We then need to find suitable $q_b$, $q_*$ and $\rho_\perp$ such that $ q_* \rho_* +  q_\perp \rho_\perp = q_0 \sigma_0 + q_1 \sigma_1$, to satisfy the correspondent operational equivalence. Note that, if we accessed perfect post-selections and preparations, then we would get $C_{S}=1$ by choosing $\sigma_1 = \ketbra{\phi}{\phi}$ and $\sigma_0 = \iden - \ketbra{\phi}{\phi}$. In practice the post-selection is not exactly projective and $\sigma_b$ will never be exactly pure, so that $C_{S} <1$ experimentally. 

We are now able to formulate a no-go theorem using this second strategy. Denoting by $p(x,y|P_*, M_F \circ M_W)$ the probability that, if the system is initialized through the preparation procedure $P_*$ and $[x|M_W]$, $[y|M_F]$ are sequentially measured one obtains outcomes $(x,y)$, we have: 
  \begin{thm}[Noncontextuality inequality based on preparation noncontextuality]
  	\label{th:preparationNC}
  	Suppose we have a noncontextual ontological model and:
  	\begin{enumerate}
	  \item \label{oe1th2}There exists a 2-outcome measurement $M_\E$ and a probability distribution $q(x)$ with median $x=0$ such that, for all $x \in \mathbb{R}$,
	     \begin{equation}
		    [x|M_W] \simeq q(x-1)[y=1|M_{\E}] + q(x)[y=0|M_{\E}].
		    \label{eq:oe1th2}
	      \end{equation}
  		\item \label{oe2th2} Given the sequential measurement $[x,y|M_F \circ M_W]$, define $[y|\tilde{M}_F]:= \int dx [x,y|M_F \circ M_W]$. Then there exists $p_d \in [0,1]$ such that
  		\begin{equation}
  		\label{eq:oe2th2}
  		[y|\tilde{M}_F] \simeq (1-p_d) [y|M_F] + p_d [y|M_D], 
  		\end{equation}
  		for some 2-outcome measurement $[y|M_D]$.
  		\item \label{oepreparation} There exists an ensemble $$\{\{q_*,P_*\},\{q_\perp,P_\perp\}\},$$ such that
  		\begin{equation}
  		\label{eq:oepreparation}
  		q_0[b=0|S] + q_1[b=1|S] \simeq  q_*P_* + q_\perp P_\perp.
  		\end{equation}
  	\end{enumerate}
   Then, if $p_-:=  \int_{-\infty}^0 p(x,y=1|P_*, M_F \circ M_W)dx$ and $p_F := p(y=1|P_*,M_F)$,
   \begin{equation}
  \label{eq:th2inequality}
  p_- \leq p_F \frac12 + (1-p_F)p_d + \frac{1-C_{S}}{2q_*}.
  \end{equation}
  \end{thm}
  
  The theorem is proved in Appendix~\ref{proof2}. It parallels Theorem~\ref{thm:1}, in particular Eq.~\eqref{eq:oe1th2} is the same as Eq.~\eqref{eq:thm1}. Because the logical structure of Appendix~\ref{proof2} parallels that of Appendix~\ref{proof1}, theorems parallel to Theorem~\ref{thm:2} and \ref{thm:3} can also be proven, with the assumption in Eq.~\eqref{eq:thm2_2}/ Eq.~\eqref{eq:thm3_2} replaced by the conjunction of Eq.~\eqref{eq:oe2th2} and Eq.~\eqref{eq:oepreparation}. 

Now that the alternative theorem is stated, let us discuss in more detail the differences with our first approach by comparing Theorem~\ref{th:preparationNC} with Theorem~\ref{thm:1}. The requirement Eq.~\eqref{eq:oe1th2} is exactly the same; the operational equivalence of Eq.~\eqref{eq:oe2th2} is strictly weaker than the correspondent Eq.~\eqref{eq:thm1_2}, since the latter requires us to verify that the weak measurement affects only slightly \emph{any} subsequent measurement, whereas the former only requires us to check the same condition for the post-selection measurement $M_F$; the operational equivalence of Eq.~\eqref{eq:oepreparation} is added, and involves the addition of preparations $S$ used for testing the quality of the post-selection, as well as of a preparation $P_\perp$ that provides a nontrivial operational equivalence; finally, the bound on $p_-$ matches the analogue one from Theorem~\ref{thm:1}, with an extra punishing term proportional to $1-C_{S}$. The bound hence becomes increasingly weak as the post-selection departs from the perfect predictability associated with projective measurements in quantum theory.

\subsection{Application: assessing current AWV experiments}\label{currexpts}

The second version of the theorem can also be compared with the quantum mechanical predictions. For example, in the unbiased noise model presented in the previous section one can show that all the operational equivalences of Theorem~\ref{th:preparationNC} are satisfied. Furthermore, one can use Eq.~\eqref{eq:quantump-} and note that $C_{S} = 1-\epsilon$ (see Appendix~\ref{quantval}). 

We can once more compare with the experimental setting of Ref.~\cite{piacentini16}. First, note that only the operational equivalences of Eq.~\eqref{eq:oe1th2} and \eqref{eq:oe2th2} are claimed, so one would need to complete this with Eq.~\eqref{eq:oepreparation} to get that the violation of the bound of Eq.~\eqref{eq:th2inequality} is a proof of contextuality. In other words, in principle the same data can be utilised by simply adding an estimation of the sharpness of the post-selection through an extra preparation satisfying Eq.~\eqref{eq:oepreparation}. We can, in fact, work out from the experimental data how close to projective the post-selection needs to be for the claim of contextuality from AWV of Ref.~\cite{piacentini16} to hold. Specifically, one has $p_d = 0.0019$, $s=8.10336$, $p_F = 0.0475865$, $\frac{p_-}{p_F} = 0.602927$ and (with the obvious choice of fair ensembles) $q_* = 1/2$. One can then estimate that Eq.~\eqref{eq:th2inequality} is satisfied only if $C_{S} > 0.996912$. We see that in this case the post-selection needs to be very close to ideal.\footnote{Not unrealistically close. A quantity comparable to $C_{S}$ was reported as $0.99709(7)$ in another contextuality experiment \cite{mazurek16}.}

\section{Remaining idealisations: perfect operational equivalences}\label{remain}

Some readers may have noticed that there is an idealisation that was not dealt with in Theorems~\ref{thm:1}-\ref{th:preparationNC}. That is, any experiment will only ever verify the operational equivalences $\simeq$  up to some approximation. Luckily, as discussed in Ref.~\cite{mazurek16}, this can be dealt with using a generic technique. One begins by assuming access to a tomographically complete set of procedures that enables the operational equivalences to be checked. The basic idea is then that whilst the ``primary'' procedures (i.e. the ones actually implemented) will not satisfy the operational equivalence exactly, we can use their statistics to find `secondary' preparations in their convex hull\footnote{In general ``supplementary procedures'' have to be implemented to ensure that the convex hull extends in all directions \cite{mazurek16}.} that do satisfy the equivalences exactly. It is to these secondary preparations that we can apply Theorems~\ref{thm:1}--\ref{th:preparationNC}. In particular, as we discussed one can apply Theorem~\ref{thm:3} both to the single qubit pointer experiment, as well as the coarse-grained version of the standard experiment -- meaning that we only need to apply the above discussion to a finite set of operational equivalences. The price for using this technique is that the secondary procedures are more mixed than the primary ones and hence will give smaller values for $C_{S}$ and $p_-$. In that sense, applying this technique builds upon the noise-robustness to non-ideal values of such parameters that we have provided here.

A last comment. The last remaining idealisation at this point is that we assumed we know a tomographically complete set of measurements. Strictly, we cannot prove that a given set of measurement procedures is complete without relying on the quantum formalism. However, one can gather evidence from the experimental data that a given set is complete. This goes beyond the scope of the present work, but is discussed in detail in Ref.~\cite{mazurek16}, and new techniques to address this issue have since been introduced in Ref.~\cite{pusey19}.  

\section{Necessity of operational equivalences for nonclassicality of AWV}
\label{sec:necessity}

We have seen that the statistics collected by the AWV experiment cannot be reproduced by a noncontextual ontological model in the presence of some extra operational constraints:
\begin{enumerate}
  \item Eqs.~\eqref{eq:thm1} and \eqref{eq:thm1_2} in the case of Theorem~\ref{thm:1}, with similar constraints for Theorems~\ref{thm:2}-\ref{thm:3}.
  \item Eqs.~\eqref{eq:oe1th2}--\eqref{eq:oepreparation} in the case of Theorem~\ref{th:preparationNC}.
\end{enumerate}
 At first sight this might sound rather involved, especially if compared to broader claims of nonclassicality of the AWV that appeared in the literature. Here, however, we show that dropping \emph{any} of the operational equivalences in any of the theorems allow for the explicit construction of a classical (noncontextual) ontological model that reproduces the anomaly. In fact, the models even reproduce the full quantum statistics of the sequential measurement on $\rho_*$ and not just the anomaly of the pointer. Hence our conditions are not only sufficient, but they are also necessary, showing that AWV can only be understood as unavoidably quantum in the presence of all the operational constraints described. Hopefully this will help in clarifying the debate that arose around this topic, showing that both `sides' are indeed correct: in a similar way in which non-local correlations are a quantum phenomenon only in a setting in which signalling has been excluded, so AWV are indeed fundamentally quantum, but only when accompanied by certain extra operational facts.
 
 \subsection{Necessity of conditions in Theorems~\ref{thm:1}-\ref{thm:3}}\label{transformation_necessity}
In both of the following models, we take the ontic state $\lambda$ to be $y$, i.e. a determination of the outcome of $M_F$, and we set
\begin{equation}
  p_{P_*}(\lambda) = p(y = \lambda|P_*,M_F).
\end{equation}
  
\subsubsection{Necessity of condition~\ref{condition:thm1}}\label{oe1th1HVM}

The basic idea of our first model is to give results for the weak measurement according to the operational distribution under the predetermined postselection \mbox{$y = \lambda$}. That is, we set
\begin{equation}
  p_{M_W}(x|\lambda) \approx p(x|P_*, \mathcal{M}^W, M_F, y=\lambda).\label{approxHVM1}
\end{equation}

Exact equality in Eq.~\eqref{approxHVM1} would allow us to reproduce the operational distribution over $x$ without any disturbance to the ontic state at all, at the price of violating the conditions on $p_{M_W}$ arising from measurement noncontextuality (a failure of condition~\ref{condition:thm1}). However, we also want to reproduce the operational fact that whether or not the weak measurement is done affects the probabilities of $M_F$ and so we add the minimal amount of disturbance necessary to achieve this. This amount of disturbance is, unsurprisingly, bounded by the $p_d$ from Eq.~\eqref{eq:thm1_2}. We then actually sample $x$ from the operational distribution for $y = \lambda'$, the disturbed ontic state, which is why Eq.~\eqref{approxHVM1} is only approximately true. The model is illustrated in Fig.~\ref{fig:hvm1}, for the full detail of how to implement the minimal disturbance see Appendix~\ref{HVMappendix}.

\begin{figure}
  \includegraphics[width=\columnwidth]{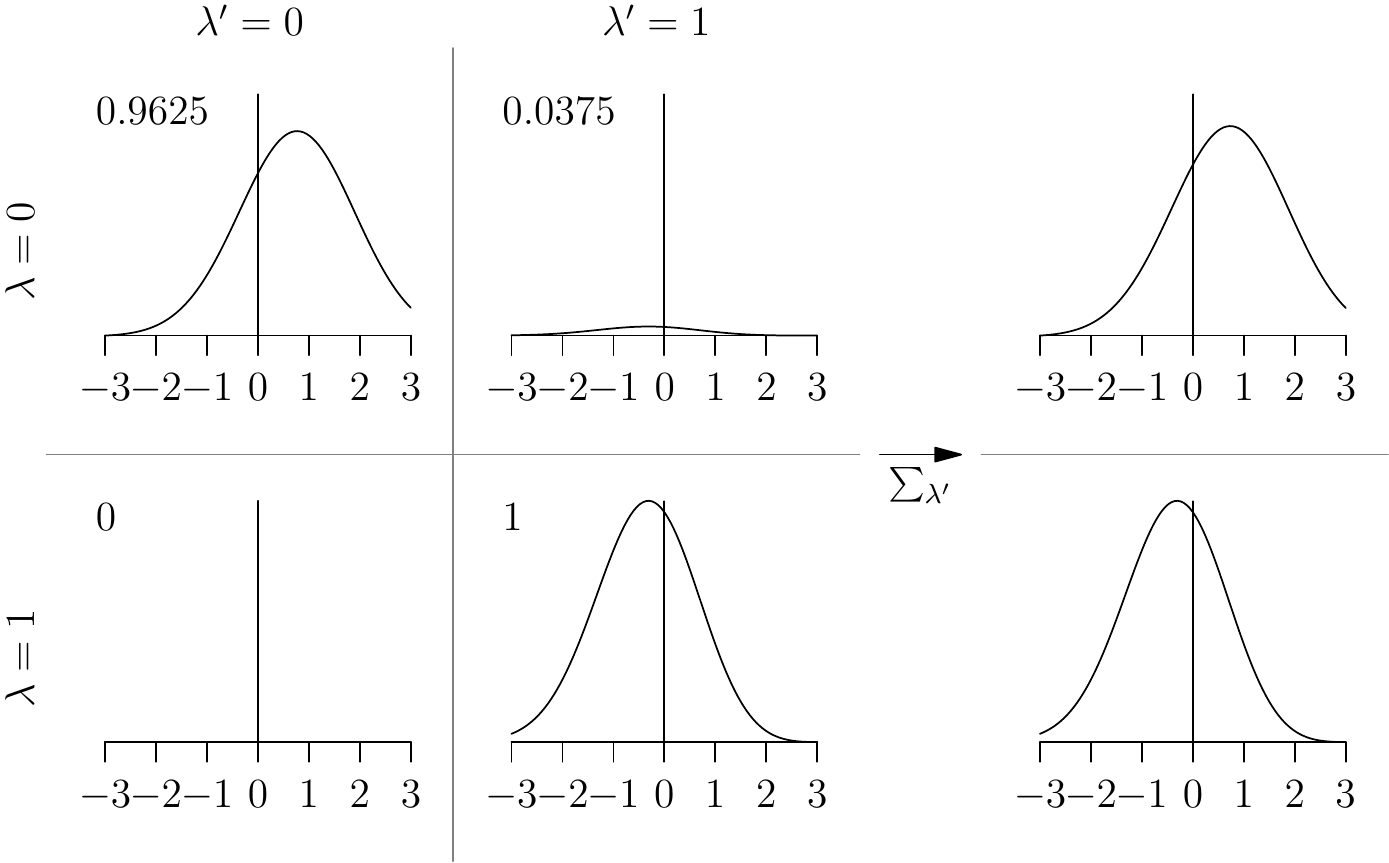}
  \caption{An illustration of the model in Sec.~\ref{oe1th1HVM}. On the left are plots of $p_{\mathcal{M}^W}(x,\lambda'|\lambda)$ against $x$, and the numbers $p_\mathcal{M}(\lambda'|\lambda) = \int_{-\infty}^\infty p_{\mathcal{M}^W}(x,\lambda'|\lambda) dx$. On the right are plots of $p_{M_W}(x|\lambda) = \sum_{\lambda'} p_{\mathcal{M}^W}(x,\lambda'|\lambda)$ against $x$. The operational probabilities used are quantum probabilities from the standard scheme with parameters chosen so that $p_F = \frac{1}{5}$, $p_d = \frac1{20}$ and ${}_{\phi}{\langle \E \rangle}_{\psi} = -\frac12$. (In particular those parameters include a rather small $s \approx 1.5$ to ensure that all features are visible. This $s$ is still large enough for our noncontextuality inequalities to be violated.) Notice on the left that $\lambda = \lambda'$ with high probability, but on the right we see the $\lambda=1$ ontic state is predisposed to give negative values of $x$.\label{fig:hvm1}}
\end{figure}

\subsubsection{Necessity of condition~\ref{condition:thm1_2}}\label{oe2th1HVM}

This time we ignore $\lambda$ and simply distribute $(x, \lambda')$ according to the operational probabilities for $(x,y)$, at the expense of a very large disturbance to the post-selection:
\begin{equation}
  p_{\mathcal{M}^W}(x,\lambda'|\lambda) = p(x,y=\lambda'|P_*, \mathcal{M}^W, M_F).
\end{equation}
By construction
\begin{equation}
  \sum_{\lambda'} p_{\mathcal{M}^W}(x,\lambda'|\lambda) = p(x|P_*, M_W),
\end{equation}
so we satisfy any operational equivalences for $M_W$ (condition~\ref{condition:thm1} is satisfied).

Intuitively, notice that $\lambda=1$ is greatly disturbed by the model since the probability of going to $\lambda' =0$ is $p(y=0|P_*,\mathcal{M}, M_F) \approx 1-p_F$ (the probability of not passing the postselection). This is a failure of condition~\ref{condition:thm1_2} whenever that probability exceeds $p_d$. These features can be seen in Fig.~\ref{fig:hvm2}.

\begin{figure}
  \includegraphics[width=\columnwidth]{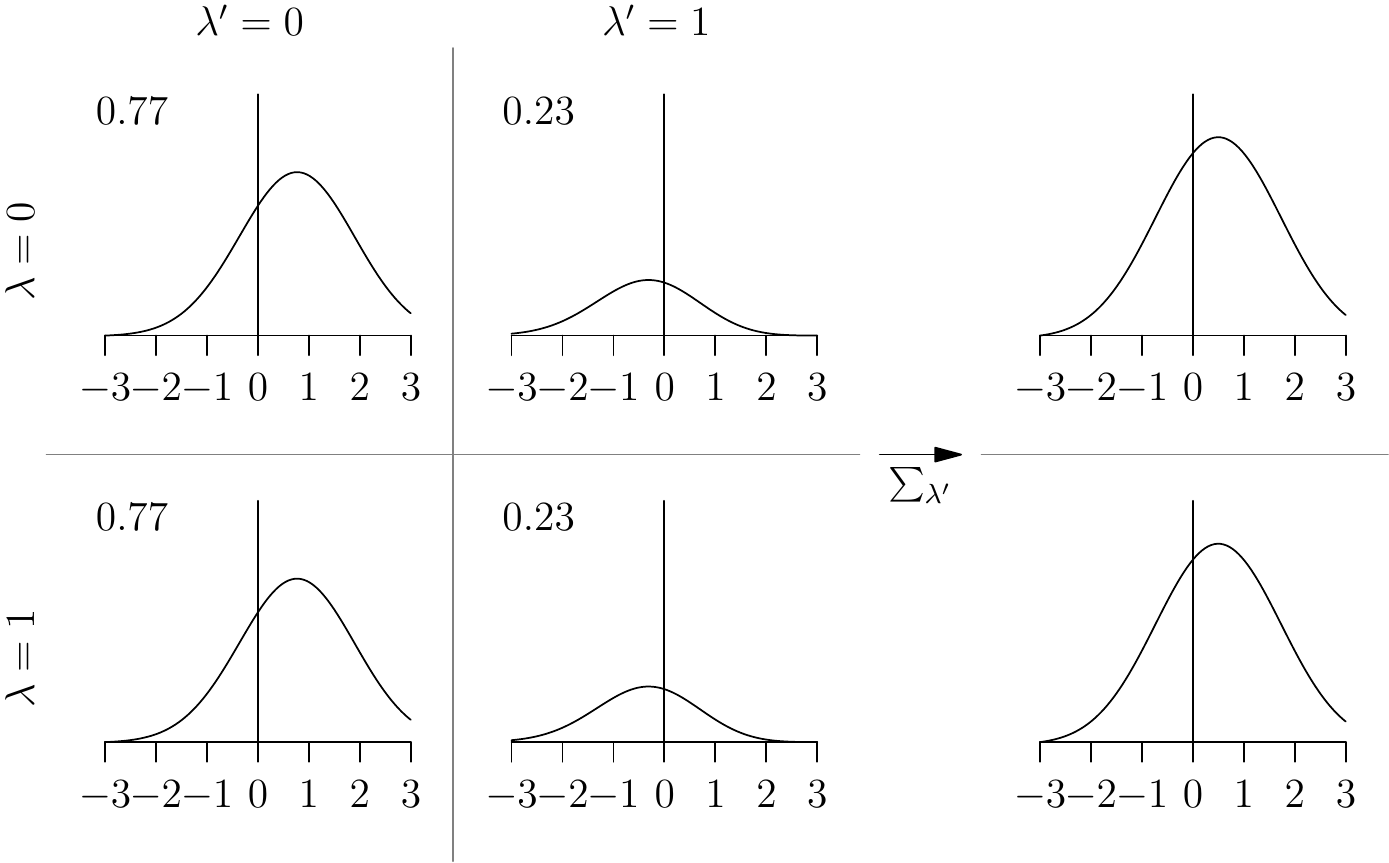}
  \caption{As in Fig.~\ref{fig:hvm1}, but for the model of Sec.~\ref{oe2th1HVM}. Notice on the right that neither ontic state is predisposed to give negative $x$, but on the left we see that the $\lambda = 1$ state is very likely to be disturbed to $\lambda'=0$.\label{fig:hvm2}}
\end{figure}

 \subsection{Necessity of conditions in Theorem~\ref{th:preparationNC}}

\subsubsection{Necessity of condition~\ref{oe1th2}}
This follows from the first model above. To satisfy condition~\ref{oepreparation}, we can set $p(\lambda|P) = p(y|P,M_F)$ for any preparation procedure $P$. This respects convexity and if two procedures are operationally equivalent they will in particular have the same $p(y|P,M_F)$ and hence the same $p(\lambda|P)$, as required by preparation noncontextuality.
	
	\subsubsection{Necessity of condition~\ref{oe2th2}}
This follows similarly from the second model above.

\subsubsection{Necessity of condition~\ref{oepreparation}}

The final ontological we consider is the $\psi$-complete model \cite{harrigan10}, which is well-known to be measurement noncontextual. In fact we will consider the generalization of the $\psi$-complete model to an arbitrary operational theory. The set of ontic states $\lambda$ is identified with the set of (convexly extremal) preparations, $p_P(\lambda) = \delta (\lambda - P)$, and the response functions are  given by the operational probabilities, $p_M(x|\lambda) = p(x|P=\lambda, M)$. This model reproduces the operational probabilities and is measurement noncontextual (that is, satisfies conditions~\ref{oe1th2} and \ref{oe2th2} of Theorem~\ref{th:preparationNC}), however it does not satisfy preparation noncontextuality, since it does not associate the same distributions to the ensembles associated to $S$ and $\{P_*, P_\perp\}$. Hence, condition \ref{oepreparation} cannot be dropped from Theorem~\ref{th:preparationNC}.

\section{Conclusions}  
Our results show that contextuality captures what is nonclassical about anomalous weak values in a way that is experimentally relevant and wide-ranging. In particular, the postselection need not be a perfect projective measurement, the pointer need not be a continuous-variable system, and if there is an imaginary part to the weak value then the real part need not be anomalous.

On the other hand, we have shown through explicit noncontextual models that if any of the operational equivalences we use are absent a classical explanation is possible. 

Our results also answer some of the questions left open in Ref.~\cite{lostaglio18}. There, it was shown that the fluctuation theorem experiments probing the Margenau-Hills work quasi probability introduced in Ref.~\cite{allahverdyan2014nonequilibrium} can witness contextuality. However, it was left open how to make the argument robust to experimental imperfections. Here we gave the tools to do so.

\begin{acknowledgments}
MP thanks Jonathan Barrett and Sina Salek for useful discussions, particularly regarding the imaginary parts of weak values, and disturbance in Ref.~\cite{karanjai15}. ML thanks Fabrizio Piacentini and Marco Genovese for useful discussions. RK thanks Elie Wolfe for discussion on the algorithmic aspects. Research at Perimeter Institute is supported by the Government of Canada through the Department of Innovation, Science and Economic Development Canada and by the Province of Ontario through the Ministry of Research, Innovation and Science. MP is supported by the Royal Commission for the Exhibition of 1851. ML acknowledges financial support from the the European Union's Marie Sklodowska-Curie individual Fellowships (H2020-MSCA-IF-2017, GA794842), Spanish MINECO (Severo Ochoa SEV-2015-0522 and project QIBEQI FIS2016-80773-P), Fundacio Cellex and Generalitat de Catalunya (CERCA Programme and SGR 875).
\end{acknowledgments}

\bibliography{refs}

\onecolumngrid

\begin{appendix}

\section{Ideal, standard quantum scenario}
\label{appendix:standard}

The channel induced by the weak measurement when the outcome is not recorded is $\mathcal{M}(\cdot) = \int_{-\infty}^{+\infty} \mathcal{M}^W_x (\cdot) =\int_{-\infty}^{+\infty} N_x(\cdot) N^\dag_x$. Using the integral $\int_{-\infty}^{+\infty}G_s(x-a)G_s(x-b)dx=e^{-(a-b)^2/4s^2}$ and Eq.~\eqref{eq:N_x}, one finds for every $\rho$
\begin{align}
\mathcal{M}(\rho) &= \mathcal{E}\rho \mathcal{E}+\mathcal{E}^{\perp}\rho\mathcal{E}^{\perp}+e^{-1/4s^2}(\mathcal{E}\rho\mathcal{E}^{\perp}+\mathcal{E}^{\perp}\rho\mathcal{E})
=\frac{1}{2}\rho+\frac{1}{2}(\mathcal{E}-\mathcal{E}^{\perp})\rho(\mathcal{E}-\mathcal{E}^{\perp})+e^{-1/4s^2}\left(\frac{1}{2}\rho-\frac{1}{2}(\mathcal{E}-\mathcal{E}^{\perp})\rho(\mathcal{E}-\mathcal{E}^{\perp})\right)\nonumber\\
&=\frac{1+e^{-1/4s^2}}{2}\rho+\frac{1-e^{-1/4s^2}}{2}(\mathcal{E}-\mathcal{E}^{\perp})\rho(\mathcal{E}-\mathcal{E}^{\perp})
=(1-p_d)\rho+p_d (\E -\mathcal{E}^{\perp})\rho(\mathcal{E}-\mathcal{E}^{\perp}), \nonumber
\end{align}
with $p_d = \frac{1-e^{-1/4s^2}}{2}$. Hence, $\mathcal{M} = p_d \mathcal{I} + (1-p_d) \mathcal{M}_D$, with $\mathcal{M}_D(\rho):= (\E - \E^\perp)\rho (\E - \E^\perp)$. It is then clear that the operational equivalences required by Theorem~\ref{thm:1} are satisfied in the ideal case.

Finally, one can compute $
p^{\rm ideal}_- = \int^0_{-\infty} dx \tr{}{ \Pi_\phi N_x \rho_* N^\dag_x} = \frac{p_F}{2} - \frac{\re\left(\braket{\Pi_\phi\E}_{\rho_*}\right)}{\sqrt{\pi}s} + o\left(\frac{1}{s}\right)$. This is a simple calculation see, e.g., the proof of Lemma~1 in Ref.~\cite{lostaglio18} (note, however, that we redefined $p^{\rm ideal}_-$ without the normalisation by the postselection probability). 
 
\section{Proof of Theorems~\ref{thm:1}-\ref{thm:3} and remarks on tightness of the inequalities}\label{proof1}
All three theorems follow from the same basic argument, hence it is convenient to formulate all of them as corollaries of the following technical lemma:
\begin{lem}[Noncontextuality inequality template 1]
	\label{lem:transformationNC}
	Suppose we have a noncontextual ontological model and that:
	\begin{enumerate}
		\item \label{oe1th1} For any input $\lambda$, the probability of a negative outcome of $[x|M_W]$ is bounded by some value independent of the ontic state:
		\begin{equation}
			\label{eq:oe1th1}
			\int_{-\infty}^0 p_{M_W}(x|\lambda)dx \leq \tilde p.
		\end{equation}
		
		\item \label{oe2th1} If $\mathcal{M}:= \int \mathcal{M}^W_x dx$, there exists $p_d \in [0,1]$ such that
		\begin{equation}
			\label{eq:oe2th1}
			\mathcal{M} \simeq (1-p_d) \mathcal{I} + p_d \mathcal{M}^D, 
		\end{equation}
		where $\mathcal{I}$ denotes the identity transformation and $\mathcal{M}^D$ some other transformation.
	\end{enumerate}
	Then, if $p_-:=  \int_{-\infty}^0 p(x,y=1| P_*, \mathcal{M}^W,M_F) dx$ and $p_F := p(y=1|P_*,M_F)$,
	\begin{equation}
		\label{eq:th1inequality}
		p_- \leq p_F \tilde p + (1-p_F)p_d =: p^{\rm NC}_{-}.
	\end{equation}
\end{lem}
\begin{proof}
		Define $\Lambda^\lambda_1 = \{ \lambda' : p_{M_F}(y=1|\lambda') \leq p_{M_F}(y=1|\lambda) \}$ (i.e. $\lambda'$ is undisturbed or uselessly disturbed, in terms of probability of passing the post-selection) and $\Lambda^\lambda_2 = \Lambda \setminus \Lambda^\lambda_1 = \{ \lambda' :p_{M_F}(y=1|\lambda') > p_{M_F}(y=1|\lambda)\}$ (i.e., $\lambda'$ usefully disturbed). In the ontological model,
		
		\begin{equation}
	p_-=  \int_{-\infty}^0 p(x,y=1|P_*, \mathcal{M}^W,M_F)dx = \int_{-\infty}^0 \int \int p_{M_F}(y=1|\lambda') p_{\mathcal{M}^W}(x,\lambda'|\lambda)p_{P_*}(\lambda) d\lambda' d\lambda dx . \label{contform}
		\end{equation}
As described in Sec.~\ref{sec:noncontextualdescription},	$p_{M_W}(x|\lambda) = \int_\Lambda p_{\mathcal{M}^W}(x,\lambda'|\lambda)d\lambda'$.
	Also, note that $\int p_{M_F}(y=1|\lambda) p_{P_*}(\lambda)d\lambda = p_F$. Hence, for the $\Lambda^\lambda_1$ part of \eqref{contform}, using Eq.~\eqref{eq:oe1th1},
	\begin{multline}
	\int_{-\infty}^0 dx \int_{\Lambda} d\lambda \int_{\Lambda^\lambda_1} d\lambda'  p_{M_F}(y=1|\lambda') p_{\mathcal{M}^W}(x, \lambda'|\lambda)p_{P_*}(\lambda)   \leq \int_{-\infty}^0 dx\int_{\Lambda}d\lambda \int_{\Lambda^\lambda_1}d\lambda'  p_{M_F}(y=1|\lambda) p_{\mathcal{M}^W}(x, \lambda'|\lambda)p_{P_*}(\lambda) \\ = \int_{-\infty}^0 dx\int_\Lambda d\lambda p_{M_F}(y=1|\lambda)p_{M_W}(x|\lambda)p_{P_*}(\lambda) - c\leq \tilde p \int_\Lambda p_{M_F}(y=1|\lambda)p_{P_*}(\lambda)d\lambda - c = \tilde p p_F - c,
	\end{multline}
	with a `correction' term  which measures the contribution to \eqref{contform} lost due to useless disturbance:
	\begin{equation}
	c= \int_{-\infty}^0 dx \int_{\Lambda} d\lambda \int_{\Lambda^\lambda_2} d\lambda' p_{M_F}(y=1|\lambda) p_{\mathcal{M}^W}(x, \lambda'|\lambda)p_{P_*}(\lambda).
	\end{equation}

	As described in Sec.~\ref{sec:noncontextualdescription}, $p_\mathcal{M}(\lambda'|\lambda) = \int_{-\infty}^\infty p_{\mathcal{M}^W}(x,\lambda'|\lambda)dx$. Hence, $p_\mathcal{M}(\lambda'|\lambda) \geq \int_{-\infty}^0 p_{\mathcal{M}^W}(x,\lambda'|\lambda)dx$. By Eq.~\eqref{eq:oe2th1} and (transformation) noncontextuality we have $p_\mathcal{M}(\lambda'|\lambda) = (1-p_d)p_\mathcal{I}(\lambda'| \lambda) + p_d p_{\mathcal{M}^D}(\lambda'|\lambda)$. Since $\mathcal{I}$ can be implemented, for example, by letting no time pass so that no dynamical evolution is possible, transformation noncontextuality requires $p_{\mathcal{I}}(\lambda'|\lambda) = \delta(\lambda'-\lambda)$. Hence, $
	\int d\lambda \int_{\Lambda^\lambda_2} d\lambda' \left(p_{M_F}(y=1|\lambda') - p_{M_F}(y=1|\lambda) \right) p_{\mathcal{I}}(\lambda'|\lambda) p(\lambda) = 0$. 
It follows that for the part of \eqref{contform} with $\lambda' \in \Lambda^\lambda_2$ we have
	\begin{align*}
	&\int_{-\infty}^0 dx \int_{\Lambda} d\lambda \int_{\Lambda^\lambda_2}d\lambda' p_{M_F}(y=1|\lambda') p_{\mathcal{M}^W}(x, \lambda'|\lambda)p_{P_*}(\lambda) \\
	&= \int_{-\infty}^0 dx \int_{\Lambda} d\lambda \int_{\Lambda^\lambda_2}d\lambda' \left(p_{M_F}(y=1|\lambda') - p_{M_F}(y=1|\lambda)\right)p_{\mathcal{M}^W}(x, \lambda'|\lambda) p_{P_*}(\lambda) + c \\
	&\leq\int_\Lambda d\lambda \int_{\Lambda^\lambda_2}d\lambda'\left(p_{M_F}(y=1|\lambda') - p_{M_F}(y=1|\lambda)\right)p_{\mathcal{M}}( \lambda'|\lambda)p_{P_*}(\lambda)+ c \\
	&= p_d  \int_\Lambda d\lambda \int_{\Lambda^\lambda_2}d\lambda'\left(p_{M_F}(y=1|\lambda') - p_{M_F}(y=1|\lambda)\right)p_{\mathcal{M}^D}( \lambda'|\lambda)p_{P_*}(\lambda)  +c \\ 
	&\leq  p_d  \int_\Lambda d\lambda \int_{\Lambda^\lambda_2}d\lambda'\left(1 - p_{M_F}(y=1|\lambda)\right)p_{\mathcal{M}^D}( \lambda'|\lambda)p_{P_*}(\lambda) + c \\
	 &\leq p_d\int_{\Lambda}\left(1 - p_{M_F}(y=1|\lambda)\right)p_{P_*}(\lambda)d\lambda + c \\
	&= (1-p_F)p_d + c.
	\end{align*}
	Summing the $\Lambda^\lambda_1$ and $\Lambda^\lambda_2$ contributions gives $p_- \leq p_F \tilde p + (1 - p_F)p_d$.
      \end{proof}

Our inequality is slightly tighter than the $p_- \leq p_F \tilde p + p_d$ one would expect from \cite{pusey14}. In order to check whether our inequality is in fact maximally tight, we applied the algorithmic approach to noncontextuality inequalities described in Ref.~\cite{schmid17}. Since that approach requires fixed operational equivalences, we repeated this procedure for many numerical values of the parameters $\tilde{p}, p_d$ and verified our inequalities define facets of the corresponding ``noncontextuality polytope'' \cite{schmid17} in each case (see Appendix~\ref{appendix:algorithmic}). It appears that our inequality is unique and tight, with the exclusion of the regime in which $p_d \geq \tilde{p}$, for which the method returns the trivial inequality $p_- \leq \tilde{p}$, which follows immediately from Eq.~\eqref{eq:oe1th1}. As we will see, in actual experiments one has $p_d \ll \tilde{p}$.

We can now prove the theorems by obtaining specific values for $\tilde p$ using noncontextuality and the operational equivalence of condition~\ref{condition:thm1} of each theorem:
\begin{proof}[Proof of Theorem~\ref{thm:1}]
  By Eq.~\eqref{eq:thm1} and measurement noncontextuality we have 
  \begin{equation}
    p_{M_W}(x|\lambda) = q(x-1)p_{M_\E}(y=1|\lambda) + q(x)p_{M_\E}(y=0|\lambda).
    \label{eq:thm1hv}
  \end{equation} Since the median of $q(x)$ is $0$ we have $\int_{-\infty}^0 q(x-1)dx \leq \int_{-\infty}^0 q(x) dx = \frac12$. In any ontological model, $\sum_y p_{M_\E}(y|\lambda) = 1$ for every $\lambda$. Integrating both sides of Eq.~\eqref{eq:thm1hv} from $-\infty$ to $0$ then gives Eq.~\eqref{eq:oe1th1} with $\tilde p = \frac12$. Hence,  we can apply Lemma~\ref{lem:transformationNC} to obtain the result.
\end{proof}

\begin{proof}[Proof of Theorem~\ref{thm:2}]
	By Eq.~\eqref{eq:thm2} and measurement noncontextuality we have 
	\begin{equation}
	p_{M_W}(\mmtm|\lambda) = p_{M_\text{triv}}(\mmtm|\lambda).
	\label{eq:thm2hv}
      \end{equation} By definition $p_{M_\text{triv}}(\mmtm|\lambda)$ is independent of $\lambda$ and $\int_{-\infty}^0 p_{M_\text{triv}}(\mmtm|\lambda)d\mmtm = \frac12$. Integrating both sides of Eq.~\eqref{eq:thm2hv} from $-\infty$ to $0$ then gives Eq.~\eqref{eq:oe1th1} with $\tilde p=\frac12$. Hence, we can apply Lemma~\ref{lem:transformationNC} to obtain the result.
\end{proof}

\begin{proof}[Proof of Theorem~\ref{thm:3}]
	By Eq.~\eqref{eq:thm3} and measurement noncontextuality we have 
	\begin{equation}
	p_{M_W}(x|\lambda) = p_m p_{M_\E}(x|\lambda) + (1-p_m) p_{M_\text{triv}}(x|\lambda).
	\label{eq:thm3hv}
      \end{equation} By definition $p_{M_\text{triv}}(x|\lambda)$ is independent of $\lambda$ and $\int_{-\infty}^0 p_{M_\text{triv}}(x|\lambda)dx = \frac12$. In any ontological model, $\int_{-\infty}^0 p_{M_\E}(x|\lambda) dx \leq \int_{-\infty}^\infty p_{M_\E}(x|\lambda) dx = 1$. Integrating both sides of \eqref{eq:thm3hv} from $-\infty$ to $0$ then gives Eq.~\eqref{eq:oe1th1} with $\tilde p = p_m + (1-p_m)\frac12 = \frac{1+p_m}2$. Applying Lemma~\ref{lem:transformationNC} gives the result.
\end{proof}

Notice that the tightness of the inequality proven in Lemma~\ref{lem:transformationNC} does not automatically imply that the inequalities in Theorems~\ref{thm:1}-\ref{thm:3} are tight, because Eqs.~\eqref{eq:thm1hv}-\eqref{eq:thm3hv} are stronger constraints than Eq.~\eqref{eq:oe1th1} with the relevant value of $\tilde p$. Since Eqs.~\eqref{eq:thm1hv} and \eqref{eq:thm2hv} reflect an infinite number of operational equivalences (one for each value of $x$), for Theorems~\ref{thm:1} and \ref{thm:2} this issue cannot be straightforwardly settled using the techniques from \cite{schmid17} alone because those only apply to finite sets of equivalences. It may be possible to gain some confidence by using a series of increasingly fine-grained but nevertheless finite operational equivalences. Theorem~\ref{thm:3} is a somewhat easier case: since it is intended to apply to a finite number of outcomes, for each number of outcomes there will in fact be a finite set of equivalences for which the relevant polytope could be calculated. In this work we leave the tightness of the inequalities in Theorems~\ref{thm:1}-\ref{thm:3} as open problems, but we find the tightness of the inequality in Lemma~\ref{lem:transformationNC} quite suggestive.

\section{Proof of Theorem~\ref{th:preparationNC} (+ extension to imaginary weak values and finite version) }\label{proof2}
We will use the same structure as in Appendix~\ref{proof1} above, with the main argument in the form a lemma.
  \begin{lem}[Noncontextuality inequality template 2]
  	\label{lem:preparationNC}
  	Suppose we have a noncontextual ontological model and:
  	\begin{enumerate}
	  \item  For any input $\lambda$, the probability of a negative outcome of $[x|M_W]$ is bounded by some value independent of the ontic state:
	 \begin{equation}
	 \label{eq:oe1pnc}
	 \int_{-\infty}^0 p_{M_W}(x|\lambda)dx \leq \tilde p.
		\end{equation}
  		\item Given the sequential measurement $[x,y|M_F \circ M_W]$, define $[y|\tilde{M}_F]:= \int dx [x,y|M_F \circ M_W]$. Then there exists $p_d \in [0,1]$ such that
  		\begin{equation}
  		\label{eq:oe2pnc}
  		[y|\tilde{M}_F] \simeq (1-p_d) [y|M_F] + p_d [y|M_D], 
  		\end{equation}
  		for some 2-outcome measurement $[y|M_D]$.
  		\item There exists an ensemble $$\{\{q_*,P_*\},\{q_\perp,P_\perp\}\},$$ such that
  		\begin{equation}
  		\label{eq:oepreparationpnc}
  		q_0[b=0|S] + q_1[b=1|S] \simeq  q_*P_* + q_\perp P_\perp.
  		\end{equation}
  	\end{enumerate}
   Then if $p_-:=  \int_{-\infty}^0 p(x,y=1|P_*, M_F \circ M_W)dx$ and $p_F := p(y=1|P_*,M_F)$,
   \begin{equation}
  \label{eq:inequalitypnc}
  p_- \leq p_F \tilde p + (1-p_F)p_d + \frac{1-C_{S}}{q_*}\max\{\tilde p-p_d,1-\tilde p\}.
  \end{equation}
  \end{lem}

  \begin{proof}
  Let us denote by $p_{S}(\lambda|b)$ the probability distribution associated to $[b|S]$. 

From the definition of an ontological model, $
C_{S} = \sum_{b,y \in \{0,1\}} \delta_{b y} \int_\Lambda d\lambda p_{M_F}(y|\lambda) q_b p_{S}(\lambda|b)$. 
From the definition of conditional probability, $q_b p_{S}(\lambda|b) = p_{S}(\lambda) p_{S}(b|\lambda)$.
 Then
\begin{align}
C_{S} &= \sum_{b,y \in \{0,1\}} \delta_{b y} \int_\Lambda d\lambda	 p_{M_F}(y|\lambda) p_{S}(b|\lambda)p_{S}(\lambda)  \leq  \int_\Lambda d\lambda	\max_{y \in \{0,1\}} p_{M_F}(y|\lambda) \sum_{b,y \in \{0,1\}} \delta_{b y}  p_{S}(b|\lambda) p_{S}(\lambda) := \int_\Lambda d\lambda \zeta(\lambda) p_{S}(\lambda) , \label{eq:correlationderivation}
\end{align}
where $\zeta(\lambda):= \max_{y \in \{0,1\}} p_{M_F}(y|\lambda)$. We now work out some inequalities that we need in order to bound $p_-$. Let us now split the set of ontological variables $\Lambda$ in the union of two disjoint sets: $\Lambda = \Lambda_0 \sqcup \Lambda_1$,
\begin{equation*}
\Lambda_0 = \{\lambda \in \Lambda| p_{M_F}(y=0|\lambda) \geq p_{M_F}(y=1|\lambda)\}, \ 
\Lambda_1 = \{\lambda \in \Lambda| p_{M_F}(y=1|\lambda) > p_{M_F}(y=0|\lambda)\}.
\end{equation*}
Note that $\Lambda_0$ ($\Lambda_1$) is the set of $\lambda$s that are more likely than not to fail (pass) the post-selection measurement.

\emph{Inequality 1}: For every $\lambda \in \Lambda$,
\begin{eqnarray}
\int_{-\infty}^{0} dx p_{M_F \circ M_W}(x, y=1|\lambda) \leq \int_{-\infty}^{0} dx p_{M_W}(x|\lambda)
 \leq \tilde p
\label{eq:inequality1derivation}
\end{eqnarray}
where we have used \eqref{eq:oe1pnc}.

\emph{Inequality 2}: For every $\lambda \in \Lambda_0$,
\begin{align}
\int_{-\infty}^{0} dx p_{M_F \circ M_W}(x, y=1|\lambda)  \leq \int_{-\infty}^{+\infty} dx p_{M_F \circ M_W}(x, y=1|\lambda) = (1-p_d) p_{M_F}(y=1|\lambda) + p_d p_{M_D}(y=1|\lambda) = \nonumber \\
(1-p_d) (1-\zeta(\lambda)) + p_d p_{M_D}(y=1|\lambda) \leq
(1-p_d) (1-\zeta(\lambda)) + p_d,
\label{eq:inequality2derivation}
\end{align}
where we used measurement noncontextuality applied to the operational equivalence of Eq.~\eqref{eq:oe2pnc} and the definition of $\zeta(\lambda)$ in $\Lambda_0$.

We can now use these inequalities to give an upper bound to $p_-$. We are going to use Eq.~\eqref{eq:inequality1derivation} for $\lambda \in \Lambda_1$ and Eq.~\eqref{eq:inequality2derivation} for $\lambda \in \Lambda_0$.
\begin{align}
	p_- = \int_{-\infty}^{0} dx p(x, y=1| P_*, M_F \circ M_W ) = 
	\sum_{i=0}^1\int_{-\infty}^{0} dx \int_{\Lambda_i} d\lambda p_{P_*}(\lambda)p_{M_F \circ M_W}(x,y=1|\lambda) \leq \nonumber \\
	(1-p_d)  \int_{\Lambda_0} d\lambda p_{P_*}(\lambda)(1- \zeta(\lambda)) + p_d  \int_{\Lambda_0} d\lambda p_{P_*}(\lambda) +
	\tilde p \int_{\Lambda_1} d\lambda p_{P_*}(\lambda).
	\label{eq:p-derivation}
	\end{align}
Let us analyse the various terms separately:
	\begin{align}
	\int_{\Lambda_0} d\lambda p_{P_*}(\lambda) & = \int_{\Lambda_0} d\lambda p_{P_*}(\lambda) p_{M_F}(y=0|\lambda) + \int_{\Lambda_1} d\lambda p_{P_*}(\lambda) p_{M_F}(y=0|\lambda) \\ & +
	\int_{\Lambda_0} d\lambda p_{P_*}(\lambda)(1-p_{M_F}(y=0|\lambda)) - \int_{\Lambda_1} d\lambda p_{P_*}(\lambda) p_{M_F}(y=0|\lambda) \nonumber \\
	& = 1-p_F + 	\int_{\Lambda_0} d\lambda p_{P_*}(\lambda)(1-\zeta(\lambda))) - \int_{\Lambda_1} d\lambda p_{P_*}(\lambda) (1-\zeta(\lambda)),
	\end{align}
where we used $\int_{\Lambda} d\lambda p_{P_*}(\lambda)p_{M_F}(y=0|\lambda) = 1-p_F$ and $p_{M_F}(y=0|\lambda) = \zeta(\lambda)$ for $\lambda \in \Lambda_0$ and $p_{M_F}(y=0|\lambda) = 1- \zeta(\lambda)$ for $\lambda \in \Lambda_1$. 
Similarly,
	\begin{align}
	\int_{\Lambda_1} d\lambda p_{P_*}(\lambda)
	=p_F - \int_{\Lambda_0} d\lambda p_{P_*}(\lambda)(1-\zeta(\lambda)) + \int_{\Lambda_1} d\lambda p_{P_*}(\lambda) (1-\zeta(\lambda))
	\label{eq:p-derivationintegral1}
	\end{align}
Substituting these in Eq.~\eqref{eq:p-derivation} we find
	\begin{eqnarray}
	p_- \leq \tilde p p_F + (1-p_F)p_d  
	+(1- \tilde p) \int_{\Lambda_0} d\lambda p_{P_*}(\lambda)(1-\zeta(\lambda)) 
	+(\tilde p-p_d) \int_{\Lambda_1} d\lambda p_{P_*}(\lambda)(1-\zeta(\lambda))\nonumber \\
	\leq \tilde p p_F + (1-p_F)p_d  
	+ \max\{\tilde p-p_d,1-\tilde p\} \int_{\Lambda} d\lambda p_{P_*}(\lambda)(1-\zeta(\lambda)).
	\end{eqnarray}

By preparation noncontextuality, Eq.~\eqref{eq:oepreparationpnc} implies $p_{S}(\lambda)=q_* p_{P_*}(\lambda) + q_\perp p_{P_\perp}(\lambda) \geq q_* p_{P_*}(\lambda)$.
Combining this with Eq.~\eqref{eq:correlationderivation}, we have
\begin{equation}
1-C_{S} = \int_\Lambda d\lambda p_{S}(\lambda) (1-\zeta(\lambda))
\geq q_* \int_\Lambda d\lambda  p_{P_*}(\lambda)(1-\zeta(\lambda))
\end{equation}

Substituting in the previous equation, we obtain the claimed bound.\end{proof}

Concerning tightness, we used the same approach as for Lemma~\ref{lem:transformationNC}, fixing the numerical values for $\tilde{p}$, $p_d$, $q_*$, $q_0$. For relevant choices of parameters we observe that the inequality defines a facet in the `non-contextuality polytope'. Furthermore, we provide numerical tools to derive all non-contextual inequalities for all choices of parameters, see Appendix~\ref{appendix:algorithmic}.

\begin{proof}[Proof of Theorem~\ref{th:preparationNC}]
  By Eq.~\eqref{eq:oe1th2} and measurement noncontextuality we have 
  \begin{equation}
    p_{M_W}(x|\lambda) = q(x-1)p_{M_\E}(y=1|\lambda) + q(x)p_{M_\E}(y=0|\lambda).
    \label{eq:oe1th2hv}
  \end{equation} Since the median of $q(x)$ is $0$ we have $\int_{-\infty}^0 q(x-1)dx \leq \int_{-\infty}^0 q(x) dx = \frac12$. In any ontological model, $\sum_y p_{M_\E}(y|\lambda) = 1$ for every $\lambda$. Integrating both sides of Eq.~\eqref{eq:oe1th2hv} from $-\infty$ to $0$ then gives Eq.~\eqref{eq:oe1pnc} with $\tilde p = \frac12$. Noting that $\tilde p = \frac12$ gives 
  \begin{equation}
    \max\{\tilde p-p_d,1-\tilde p\} = \max \left\{ \frac12 - p_d, \frac12\right\} = \frac12,
  \end{equation}
  we can apply Lemma~\ref{lem:preparationNC} to obtain the result.
\end{proof}

The extensions to imaginary weak values and to finite versions can be easily derived from Lemma~\ref{lem:preparationNC} following the same procedure as at the end of Sec.~\ref{proof1}. The situation regarding tightness of the inequalities also mirrors the discussion there.

\section{Noisy implementation of the weak value}\label{quantval}
\begin{lem}
	
	In quantum theory, a weak measurement of the projector $\E$ with initial spread of the pointer $s$ and imperfect postselection of $\Pi_\phi$ with $\epsilon-$unbiased noise as in Eq.~\eqref{eq:noisypostselection-main} achieves 
	\begin{equation}
	\label{eq:p-noisyappendix}
	p^{\rm noisy}_- = \frac{p_F}{2}  - \frac{1}{\sqrt{\pi}s}\re\left(\braket{[y=1|M_F] \E}_{\rho_*}   \right)   + o\left(\frac{1}{s}\right) .
	\end{equation}
	where $C_{S} = 1-\epsilon$. The operational equivalences required by Theorem~\ref{thm:1} are satisfied, and those of Theorem~\ref{th:preparationNC} can be satisfied by introducing the preparations ($d \equiv\Tr[\iden]$)
	\begin{equation}
	\sigma_0 = \frac{\iden-\Pi_\phi}{d-\Tr\Pi_\phi}, \quad \sigma_1 = \frac{\Pi_\phi}{\Tr\Pi_\phi}.
	\end{equation}
\end{lem}
Note that the preparations $[b|S]$ were taken to have singular density operators, but this assumption does not imply an extra idealization. In fact, if we add unbiased noise to $S$, $\sigma_1 = (1-\delta) \frac{\Pi_\phi}{\Tr\Pi_\phi} + \delta \frac{\iden}{d}$ and similarly for $\sigma_0$, we could absorb $\delta$ by a redefinition of $\epsilon$. Also note that exactly the same proof shows that the operational equivalences required by Theorems~\ref{thm:2}-\ref{thm:3}, as well as for the imaginary weak values and finite versions of Theorem~\ref{th:preparationNC}, do hold. Finally, for the imaginary weak value version, $p_-$ has a similar expression as Eq.~\eqref{eq:p-noisyappendix}, but involving the imaginary part of the weak value.

\begin{proof}
	The weak measurement scheme with $\epsilon$-unbiased noise in the post-selection coincides with the standard scheme described in Sec.~\ref{stage} with the only difference that the postselection is taken to be
	\begin{equation*}
	\{[y=1|M_F],[y=0|M_F]\} = (1- 2\epsilon) \{\Pi_\phi,\iden- \Pi_\phi\} + 2\epsilon \{\iden/2, \iden/2\}.
	\end{equation*}
	Concerning the relation between $C_{S}$ and $\epsilon$:
	\begin{eqnarray}
	C_{S} &=& q_0p(y=0|b=0,S,M_F)+q_1p(y=1|b=1,S,M_F)\nonumber\\
	&=&q_0((1-2\epsilon)p(y=0|b=0,S,\{\Pi_\phi,\iden-\Pi_\phi\})+\epsilon)
	+q_1((1-2\epsilon)p(y=1|b=1,S,\{\Pi_\phi,\iden-\Pi_\phi\})+\epsilon)\nonumber\\
	&=&q_0(1-\epsilon)+q_1(1-\epsilon)=1-\epsilon.
	\end{eqnarray}
	\emph{Operational equivalences:} The operational equivalences of Theorem~\ref{thm:1} are satisfied by following the same argument as described in the main text for the ideal case, since none of them involve the postselection.
	
	Concerning the equivalences required for Theorem~\ref{th:preparationNC} and related imaginary/finite versions, the ones that do not follow immediately from previous arguments are Eq.~\eqref{eq:oe2th2} and Eq.~\eqref{eq:oepreparation}.
	
	To prove Eq.~\eqref{eq:oe2th2} we can start with the definition 
\begin{equation}
	[y|\tilde{M}_F]:= \int^{+\infty}_{-\infty} [x,y|M_F \circ M_W] = \int^{+\infty}_{-\infty} dx N^\dag_x [y|M_F] N_x = \mathcal{M}^\dag ([y|M_F]),
\end{equation}	
and, using Eq.~\eqref{eq:M_opeq}, obtain
	\begin{equation}
	[y=1|\tilde{M}_F] = \mathcal{M}^\dag([y=1|M_F]) = (1-2\epsilon) \mathcal{M}^\dag (\Pi_\phi) + \epsilon \iden = p_d [y=1|M_F] + (1-p_d)[(1-2\epsilon)\mathcal{M}^\dag_D(\Pi_\phi) + \epsilon \iden].
	\end{equation}
	By defining a POVM $\{M_D,\iden - M_D\}$ with $M_D = (1-2\epsilon)\mathcal{M}^\dag_D(\Pi_\phi) + \epsilon \iden$, we can see that Eq.~\eqref{eq:oe2th2} is satisfied with the same $p_d$ as in the ideal case, $p_d = (1- e^{-1/4s^2})/2$. 
	
	Moving on to Eq.~\eqref{eq:oepreparation}, to satisfy it we need a careful choice of $P_\perp$ with the aim of maximising $q_*$ and hence the violation. We will leave $q_*$ as a free parameter, but note that a choice satisfying Eq.~\eqref{eq:oepreparation} always exists for any choice of $P_*$ given by $\rho_*$:
	$$ q_* = 1/d, \; \; \rho_\perp = \frac{\iden - \rho_*}{d-1}, \; \; q_1=\frac{\Tr[\Pi_\phi]}{d}.$$
	In fact, with these choices, 
	\begin{equation}
	q_* \rho_* + q_\perp\rho_\perp = q_0 \sigma_0 + q_1 \sigma_1 = \iden/d.
	\end{equation}
	\emph{Expression for $p^{\rm noisy}_-$:} 
	for both the definition of $p_-$ of Theorem~\ref{thm:1} and that of Theorem~\ref{th:preparationNC}, using Eq.~\eqref{eq:N_x}:
	\begin{equation}
	p^{\rm noisy}_- = \epsilon \int_{-\infty}^0 dx\Tr(N^\dag_x N_x \rho_*) + (1-2\epsilon) \int_{-\infty}^0 \Tr(N^\dag_x \Pi_\phi N_x \rho_*)dx.
	\end{equation}
	
	For the first term, since $N^\dag_x N_x = G_s^2(x-1) \E+ G_s^2(x)\E^\perp$, using the integral $
	\int_{-\infty}^{0}dx G^2_s(x-1)= \frac{1}{2}\erfc\left(\frac{1}{s}\right)$ expressed using the complementary error function $\erfc(x)\equiv 1-\erf(x)\equiv 1-\frac{1}{\sqrt{\pi}}\int_{-x}^xe^{-t^2}dt$ and the expansion $ {\rm erfc}(1/s) = 1 - 2/(\sqrt{\pi}s)+ o(1/s)$,
	\begin{equation}
	\int_{-\infty}^0 dx\Tr(N^\dag_x N_x \rho_*) = \frac{1}{2} \erfc\left(\frac{1}{s}\right) p_\E + \frac{1}{2}(1-p_\E) = \frac{1}{2} - \frac{p_\E}{\sqrt{\pi}s} + o\left(\frac{1}{s}\right).  
	\end{equation}
	where $p_\E = \Tr(\E \rho_*)$.
	
	For the second term, from Eq.~\eqref{eq:N_x} and the integral $
	\int_{-\infty}^{0}G_s(x-1)G_s(x)dx=\frac{e^{-1/4s^2}}{2}\erfc\left(\frac{1}{2s}\right)$ we get
	\begin{align}
	\label{eq:perfectintegral}
	\int_{-\infty}^0 \Tr(N^{s \dag}_x \Pi_\phi N_x \rho_*) dx = &  \frac{1}{2} \erfc\left(\frac{1}{s}\right)    \Tr(\mathcal{E} \Pi_\phi \mathcal{E} \rho_*) 
	+ \frac{e^{-1/(4s^2)}}{2}  \erfc\left(\frac{1}{2s}\right) \Tr((\mathcal{E}^\perp \Pi_\phi \mathcal{E} + \mathcal{E} \Pi_\phi \mathcal{E}^\perp)\rho_*)  
	+ \frac{1}{2}\Tr(\mathcal{E}^\perp \Pi_\phi \mathcal{E}^\perp \rho_*) \nonumber
	\\ = & \frac{1}{2}\Tr(\Pi_\phi \rho_*) - \frac{1}{2 \sqrt{\pi} s} \Tr((\Pi_\phi \E + \E \Pi_\phi)\rho_*)  +   o\left(\frac{1}{s}\right),   
	\end{align}
	and we note that $\Tr\left((\Pi_\phi \E + \E \Pi_\phi)\rho_*\right) = 2\re\left( \braket{\Pi_\phi \E}_{\rho_*} \right)$.
	Substituting everything into the expression for $p^{\rm noisy}_-$,
	\begin{equation}
	p^{\rm noisy}_- = \frac{p_F}{2}  - \frac{1}{\sqrt{\pi}s}\re\left(\braket{[y=1|M_F]\E}_{\rho_*}\right)   + o\left(\frac{1}{s}\right) .
	\end{equation}
\end{proof}

\section{Measurements of pointer momentum}\label{imaginaryappendix}
%

Now we calculate $p_-$, the probability of a negative value of $p$ under the postselection. For simplicity we will only consider the ideal case, where $[y=1|M_F]$ is a projection $\Pi_\phi$. However, the noisy case can be derived extending the treatment below in the same way as we did with the position measurement of the pointer in Appendix~\ref{quantval}. Thus,
\begin{equation}
  p_- = \int_{-\infty}^{0} \Tr(N_\mmtm^\dagger \Pi_\phi N_\mmtm \rho_*) d\mmtm = \frac12\left(\Tr(\E \Pi_\phi \E \rho_*) + \Tr(\E^\perp \Pi_\phi \E^\perp \rho_*) + \alpha\Tr(\E^\perp \Pi_\phi \E \rho_*) + \alpha^*\Tr(\E \Pi_\phi \E^\perp \rho_*)\right)
\end{equation}
with integral (recalling Eq.~\eqref{eq:gaussianpointer})
\begin{equation}
  \alpha = 2\int_{-\infty}^{0}\abs{\braket{\mmtm|\Psi}_P}^2\exp(-i\mmtm)d\mmtm = \exp\left( -\frac1{4s^2} \right) \left(1 + \erf\left( \frac{i}{2s} \right)\right),
\end{equation}
To calculate $\alpha^*$ recall that the $\erf$ of a purely imaginary number is purely imaginary. Using $\alpha \approx 1 + \frac{i}{\sqrt{\pi}s}$ and $\Tr((\E + \E^\perp)\Pi_\phi(\E + \E^\perp)\rho_*) = \Tr(\Pi_\phi \rho_*) = p_F$ we find, at leading order in $1/s$,
\begin{equation}
  p_- \approx \frac{p_F}2 + \frac{1}{\sqrt{\pi}s}\re(i\Tr(\E^\perp \Pi_\phi \E \rho_*)).
\end{equation}
Since $\E^\perp = \iden - \E$ and $\im(\Tr(\E \Pi_\phi \E \rho_*)) = 0$ this gives at leading order in $1/s$,
\begin{equation}
  p_- \approx \frac12 - \frac{\im\left(\braket{\Pi_\phi\E}_{\rho_*}\right)}{\sqrt{\pi}s}.
\end{equation}

\section{Qubit pointers}\label{qubitappendix}
In Ref.~\cite{wu09} weak measurements using qubit pointers are constructed, with the weakness controlled by a parameter in the interaction between the system and pointer. It turns out that, as in the continuous pointer case, one can also use a fixed interaction and control the weakness using a parameter in the pointer state. For consistency we take that approach here.

The interaction we consider is $U = \E \otimes Z + \E^\perp \otimes \iden$ where $Z$ denotes the Pauli-$Z$ operator on the qubit pointer. This interaction is basically a controlled-phase gate where the control is $\E$ versus $\E^\perp$. Indeed, by preparing the pointer in $\ket{X=-1} = \frac{1}{\sqrt2}(\ket{0} - \ket{1})$ and measuring Pauli-$X$ on the pointer one can carry out a strong measurement of $\E$ with the usual disturbance. On the other hand, since $Z\ket{0} = \ket{0}$, preparing the pointer in $\ket{0}$ would mean $U$ acts as identity on the system and hence causes no disturbance. This suggests we can achieve a weak measurement by taking an initial pointer state of $\ket{\Psi_\epsilon} = \cos\epsilon\ket{0} - \sin\epsilon\ket{1}$, where $\epsilon$ is small. Measuring $X$ on the pointer gives Kraus operators (here $\ket{X=1} = \frac{1}{\sqrt2}(\ket{0} + \ket{1})$)
\begin{equation}
N_{\pm 1} = \braket{X = \pm 1|U|\Psi_\epsilon} = \frac1{\sqrt 2}\left(\E(\cos\epsilon \pm \sin\epsilon) + \E^\perp(\cos\epsilon \mp \sin\epsilon)\right) = \frac{1}{\sqrt 2}\left( \cos\epsilon \, \iden \pm \sin\epsilon (\E-\E^\perp) \right),
\end{equation}
and hence POVM elements
\begin{equation}
  N_{\pm 1}^\dagger N_{\pm 1} = \frac{\iden}{2}   \pm \cos\epsilon\sin\epsilon (\E - \E^\perp),
\end{equation}
so that
\begin{equation}
  N_{+1}^\dagger N_{+1} = (1-p_m)\frac\iden2 + p_m\E, \quad \quad   N_{-1}^\dagger N_{-1} = (1-p_m)\frac\iden2 + p_m\E^\perp,
\end{equation}
where $p_m = 2\cos\epsilon\sin\epsilon = \sin(2\epsilon)$. Hence, Eq.~\eqref{eq:thm3} is satisfied.

If we ignore the outcome of the measurement on the pointer then we apply a channel
\begin{equation}
  \mathcal{M}(\rho) = N_{+1} \rho N_{+1}^\dagger + N_{-1} \rho N_{-1}^\dagger = \cos^2 \epsilon \rho + \sin^2 \epsilon (\E - \E^\perp)\rho(\E - \E^\perp) = (1-p_d)\rho + p_d \mathcal{M}^D (\rho),
\end{equation}
where $p_d = \sin^2 \epsilon$ and $\mathcal{M}^D(\rho) = (\E - \E^\perp)\rho(\E - \E^\perp)$, Eq.~\eqref{eq:thm3_2} is satisfied.

Finally, considering a perfect post-selection onto a projector $\Pi_\phi$, we can calculate
\begin{multline}
  p_- = \Tr(N_{-1}^\dagger \Pi_\phi N_{-1} \rho_*) =\\ \frac12(\cos^2\epsilon\Tr(\Pi_\phi\rho_*) + \sin^2\epsilon\Tr((\E - \E^\perp)\Pi_\phi(\E - \E^\perp)\rho_*) - \sin\epsilon\cos\epsilon\Tr((\E - \E^\perp)\Pi_\phi\rho_* + \Pi_\phi(\E - \E^\perp)\rho_*).
\end{multline}
Expanding to first order in $\epsilon$ gives
\begin{equation}
  p_- \approx \frac{p_F}2 - \frac\epsilon2\Tr((\E - \E^\perp)\Pi_\phi\rho_* + \Pi_\phi(\E - \E^\perp)\rho_*) = \frac{p_F}2 - \epsilon\re(\Tr(\Pi_\phi(\E - \E^\perp)\rho_*)) = \frac{p_F}2 - \epsilon(2\re(\Tr(\Pi_\phi\E\rho_*)) - p_F),
\end{equation}
and since $p_m \approx 2\epsilon$ we obtain, at leading order in $\epsilon$,
\begin{equation}
  p_- \approx  p_F \frac{1+p_m}{2} - 2\epsilon\re\left(\braket{\Pi_\phi\E}_{\rho_*}\right).
\end{equation}

\section{Details of minimal-disturbance ontological model}\label{HVMappendix}
The weak measurement $\mathcal{M}^W$ disturbs the system so that the operational probabilities for the post-selection following it, $p(y|P_*,\mathcal{M},M_F)$ differ from those that would be obtained without the weak measurement, $p(y|P_*,M_F)$. Normally the post-selection becomes slightly more likely, i.e. $\epsilon := p(1|P_*,\mathcal{M},M_F) - p(1|P_*,M_F) > 0$, because the post-selection is chosen almost orthogonal to the preparation and the weak measurement makes the state of the system slightly mixed. We will construct a model under this assumption, but if the opposite is true then we simply need to exchange the roles of $y=0$ and $y=1$ in the rest of the discussion. By normalization $\epsilon = p(0|P_*,M_F) - p(0|P_*,\mathcal{M},M_F)$, and clearly $\epsilon \leq 1$ (indeed $\epsilon$ is just the total variation distance between $p(y|P_*,M_F)$ and $p(y|P_*,\mathcal{M},M_F)$). Hence we can define
\begin{equation}
D(y'|y) = \delta_{y'y} + \frac{\epsilon}{p(0|P_*,M_F)}S(y'|y).
\end{equation}
\begin{equation}
S(y'|y) = \begin{cases}-1 & y=0,y'=0 \\ 1 & y=0, y'=1 \\ 0 & y=1\end{cases}.
\end{equation}
This is a ``minimally disturbing'' \cite{villani03} conditional distribution such that
\begin{equation}
p(y'|P_*,\mathcal{M},M_F) = \sum_{y} D(y'|y)p(y|P_*,M_F).
\end{equation}
We use this disturbance in the representation of $\mathcal{M}$ in the ontological model:
\begin{equation}
p_{\mathcal{M}^W}(x,\lambda'|\lambda) = p(x|P_*,b=1,\mathcal{M},M_F,y=\lambda')D(y'=\lambda'|y=\lambda).
\end{equation}

By construction
\begin{equation}
p_{\mathcal{M}}(\lambda'|\lambda) = \int_{-\infty}^\infty p_{\mathcal{M}^W}(x,\lambda'|\lambda)dx = D(y'=\lambda'|y=\lambda),
\end{equation}
and we have that
\begin{equation}
D(y'|y) = (1-p_d)\delta_{y'y} + p_d \left(\delta_{y'y} + \frac{\epsilon}{p(0|P_*,M_F)p_d}S(y'|y)\right)
\end{equation}
which suggests that in order to satisfy condition~\ref{condition:thm1_2} of Theorems~\ref{thm:1}-\ref{thm:3} we should set
\begin{equation}
p_{\mathcal{M}^D}(\lambda'|\lambda) = \delta_{\lambda'\lambda} + \frac{\epsilon}{p(0|P_*,M_F)p_d}S(y'=\lambda'|y=\lambda).
\end{equation}
It is easy to see that this is normalized and is clearly positive except perhaps for
\begin{equation}
p_{\mathcal{M}^D}(\lambda'=0|\lambda=0) = 1 - \frac{\epsilon}{p(0|P_*,M_F)p_d},
\end{equation}
which is positive provided $p_d \geq \frac{\epsilon}{p(0|P_*,M_F)}$. To check this we note that the operational equivalence of condition~\ref{condition:thm1_2} on $\tilde{M}_F$ tells us that
\begin{equation}
p(1|P_*,\mathcal{M},M_F) = (1-p_d)p(1|P_*,M_F) + p_d p(y=1|P_*,\mathcal{M}^D,M_F)
\end{equation}
so that, since $p(y=1|P_*,\mathcal{M}^D,M_F) \leq 1$,
\begin{equation}
\frac{\epsilon}{p(0|P_*,M_F)} = \frac{p(1|P_*,\mathcal{M},M_F) - p(1|P_*,M_F)}{1-p(1|P_*,M_F)} = p_d\frac{p(y=1|P_*,\mathcal{M}^D,M_F) - p(1|P_*,M_F)}{1-p(1|P_*,M_F)} \leq p_d.
\end{equation}
as required.

\section{Algorithmic approach to tightness}
\label{appendix:algorithmic}

We discretize the problem and use the algorithmic approach of Ref.~\cite{schmid17}, to which we refer for extra details, in order to verify that the noncontextuality inequalities of Lemmas~\ref{lem:transformationNC} and \ref{lem:preparationNC} are indeed facet inequalities of the noncontextuality polytope describing the relevant statistics. We first set up the general algorithmic problem and then see how to apply to each theorem.
\subsection{Setting up the problem}\label{algsetup}
Since we will be dealing with arrays of procedures it will be useful to number them as follows: 
\begin{equation}
P_\perp \leftrightarrow P_1, \quad P_* \leftrightarrow P_2, \quad [b=0|S]  \leftrightarrow P_3, \quad [b=1|S] \leftrightarrow P_4
\end{equation}

The operational equivalence of Eq.~\eqref{eq:oepreparationpnc} can thus be written as 
\begin{equation}
  q_\perp P_1+(1-q_\perp) P_2\simeq q_0 P_3+(1-q_0) P_4.\label{eq:algoprep}
\end{equation}

Since the definition of $p_-$ and the relevant constraints only involve a coarse graining of the measurement outcome of $M_W$ (the weak measurement), we denote a binary-outcome coarse-graining of $M_W$ as
\begin{equation}
[X=-1|M_W^{\rm bin}] =
\int_{-\infty}^0 {\rm d}x [x|M_W], \quad [X=+1|M_W^{\rm bin}] = \int_0^{\infty}{\rm d}x [x|M_W].
\end{equation}

Henceforth, we will consider the sequential measurement $M_F\circ M_W^{\rm bin}$ rather than $M_F\circ M_W$. The operational equivalence of Eq.~\eqref{eq:oe2pnc} used in Lemma~\ref{lem:preparationNC} is
\begin{equation}
\label{eq:algMFoe}
	[y|\tilde M_F] = \sum_{X=\pm1} [X,y|M_F\circ M_W^{\rm bin}]\simeq (1-p_d)[y|M_F]+p_d[y|M_D].
\end{equation}
Finally Eq.~\eqref{eq:oe1pnc} (which appears in both lemmas) becomes the condition  
\begin{equation}
\label{eq:algoptilde}
  p_{M_W^{\rm bin}}(X=-1|\lambda)\leq \tilde p\in [0,1]\quad \forall \lambda\in\Lambda.
\end{equation}

Similarly we number the relevant measurements as $\{M_1, M_2, M_3\}$ and their outcomes by $m\in\{1,2,3,4\}$, defining events $[m|M_i]$ as 
\begin{eqnarray}
	M_1&:& [1|M_1] = [1|M_F], [2|M_1]= [0|M_F], \nonumber\\
	M_2&:& [1|M_2]= [1|M_D], [2|M_2]= [0|M_D], \nonumber\\
	M_3&:& [1|M_3]= [X=-1,y=1|M], [2|M_3]= [X=-1,y=0|M],\nonumber\\
	&& [3|M_3]=[X=+1,y=1|M], [4|M_3]=[X=+1,y=0|M].\nonumber\\
\end{eqnarray}
The operational equivalence of Eq.~\eqref{eq:algMFoe} can then be restated as
\begin{eqnarray}
	(1-p_d)[1|M_1]+p_d[1|M_2]&\simeq& [1|M_3]+[3|M_3], \nonumber \\
	(1-p_d)[2|M_1]+p_d[2|M_2]&\simeq& [2|M_3]+[4|M_3],
	\label{eq:algomeasurement}
\end{eqnarray}
whilst Eq.~\eqref{eq:algoptilde} becomes
\begin{equation}
  p_{M_3}(1|\lambda) + p_{M_3}(2|\lambda) \leq \tilde p\in [0,1]
  \label{eq:algoptilde2}
\end{equation}

Applying measurement non-contextuality to Eq.~\eqref{eq:algomeasurement} gives two linear constraints on the $P_{M_i}$, on top of which we have \eqref{eq:algoptilde2}, normalization, and positivity.

For any fixed $\lambda$, we can see an assignment of the $p_{M_i}(m|\lambda)$ as a $8$-component vector. The set of all assignments compatible with the above constraints defines a polytope in this space, which we denote as \textsf{weakvaluespolysymbN} in the accompanying code \cite{code}. Its vertices will be denoted by $\kappa$. The vertex assignments in the polytope are denoted by $p_{M_i}(m|\kappa)$. For every $\lambda$, we can decompose $p_{M_i}(m|\lambda)$ as 
\begin{equation}
p_{M_i}(m|\lambda) = \sum_{\kappa} w(\kappa|\lambda)p_{M_i}(m|\kappa),
\end{equation}
where $w(\kappa|\lambda) \geq 0$, $\sum_{\kappa} w(\kappa|\lambda) =1$. Hence, we can characterise all possible assignments by computing the vertex assignments. Doing the vertex enumeration with SageMath we find there are $16$ such vertices, $\kappa_1,\dots, \kappa_{16}$.

\subsection{Tightness of the inequality in Lemma~\ref{lem:transformationNC}}

Let us consider inequality in Lemma~\ref{lem:transformationNC}. Ref.~\cite{schmid17} does not consider transformation noncontextuality, and it is not obvious how to extend the approach there to transformation noncontextuality in general. But for checking tightness in our scenario it happens that we do not require such an extension. We will prove the following result, showing that a transformation and measurement non-contextual model for the weak value experiment exists if there exists a model satisfying the original assumptions of Ref.~\cite{pusey14} -- i.e. measurement non-contextuality and outcome determinism. 
\begin{lem}
  Suppose there exists a given model which satisfies Eq.~\eqref{eq:oe1th2}, is measurement noncontextual for the equivalence Eq.~\eqref{eq:oe2th2}, and represents $M_F$ for all $\lambda$ with $p_{M_F}(y|\lambda) \in \{0, 1\}$. Then there exists a derived model which satisfies Eq.~\eqref{eq:oe1th1}, is transformation noncontextual for the equivalence Eq.~\eqref{eq:oe2th1}, and makes the same operational predictions as the given model.
\end{lem}

\begin{proof}
The derived model, probabilities of which we denote using $\mathfrak{p}$, will take the ontic state $\lambda$ to be determination of $y$ (as in Sec.~\ref{transformation_necessity}). In fact, it is constructed from the given model, which we denote by the usual $p$, by coarse-graining together all ontic states that assign the same outcome $y$ to $M_F$. In particular we set
\begin{equation}
  \mathfrak{p}_{P_*}(\lambda = y) := \int_\Lambda p_{M_F}(y|\lambda) p_{P_*}(\lambda)d\lambda,
\end{equation}
so that we have the same predictions for an immediate measurement of $M_F$. We also set
\begin{equation}
  \mathfrak{p}_{\mathcal{M}^W}(x,\lambda'=y'|\lambda=y) := \frac1{\mathfrak{p}_{P_*}(\lambda = y)}\int_\Lambda p_{M_F \circ M_W}(x,y'|\lambda) p_{M_F}(y|\lambda) p_{P_*}(\lambda)d\lambda.\label{newpw}
\end{equation}
This gives
\begin{equation}
  \sum_y \mathfrak{p}_{\mathcal{M}^W}(x,\lambda' = y'|\lambda=y)\mathfrak{p}_{P_*}(\lambda=y) = \int_\Lambda p_{M_F \circ M_W}(x,y'|\lambda) p_{P_*}(\lambda)d\lambda,
\end{equation}
so that we also have the same predictions for $\mathcal{M}^W$ followed by $M_F$. From Eq.~\eqref{newpw}, we can calculate
\begin{equation}
  \mathfrak{p}_{M_W}(x|\lambda = y) = \sum_{y'} \mathfrak{p}_{\mathcal{M}^W}(x,\lambda'=y'|\lambda=y) = \frac1{\mathfrak{p}_{P_*}(\lambda = y)}\int_\Lambda p_{M_W}(x|\lambda) p_{M_F}(y|\lambda) p_{P_*}(\lambda)d\lambda,
\end{equation}
and hence, since the given model satisfies Eq.~\eqref{eq:oe1th2},
\begin{multline}
  \int_{-\infty}^0 \mathfrak{p}_{M_W}(x|\lambda = y) dx = \frac1{\mathfrak{p}_{P_*}(\lambda = y)}\int_\Lambda \left(\int_{-\infty}^0 p_{M_W}(x|\lambda)dx\right) p_{M_F}(y|\lambda) p_{P_*}(\lambda)d\lambda \\ \leq \tilde{p} \frac1{\mathfrak{p}_{P_*}(\lambda = y)}\int_\Lambda  p_{M_F}(y|\lambda) p_{P_*}(\lambda)d\lambda = \tilde{p},
\end{multline}
giving Eq.~\eqref{eq:oe1th1} as claimed. Finally, we can calculate
\begin{equation}
  \mathfrak{p}_{\mathcal{M}}(\lambda' = y'|\lambda = y) = \int_{-\infty}^\infty \mathfrak{p}_{\mathcal{M}^W}(x,\lambda'=y'|\lambda=y) dx = \frac1{\mathfrak{p}_{P_*}(\lambda = y)}\int_\Lambda p_{\tilde{M_F}}(y'|\lambda) p_{M_F}(y|\lambda)  p_{P_*}(\lambda)d\lambda.
\end{equation}
Then since the given model is measurement noncontextual for Eq.~\eqref{eq:oe2th2} we find
\begin{multline}
  \mathfrak{p}_{\mathcal{M}}(\lambda' = y'|\lambda = y) = \frac1{\mathfrak{p}_{P_*}(\lambda = y)}\left((1-p_d)\int_\Lambda p_{M_F}(y'|\lambda) p_{M_F}(y|\lambda) p_{P_*}(\lambda)d\lambda + p_d\int_\Lambda p_{M_D}(y'|\lambda) p_{M_F}(y|\lambda) p_{P_*}(\lambda)d\lambda \right) \\= (1-p_d)\delta_{y'y} + p_d \mathfrak{p}_{\mathcal{M}^D}(\lambda'=y'|\lambda = y),
\end{multline}
where for the first term we have used outcome determinism to find $p_{M_F}(y'|\lambda) p_{M_F}(y|\lambda) = \delta_{y' y'} p_{M_F}(y|\lambda)$ and in the second we have defined
\begin{equation}
  \mathfrak{p}_{\mathcal{M}^D}(\lambda'=y'|\lambda = y) := \frac1{\mathfrak{p}_{P_*}(\lambda = y)}\int_\Lambda p_{M_D}(y'|\lambda) p_{M_F}(y|\lambda) p_{P_*}(\lambda)d\lambda.
\end{equation}
Hence we satisfy transformation noncontextuality for Eq.~\eqref{oe2th1}.
\end{proof}
We believe the converse also holds but we do not strictly require that here, since we have already proven that our inequality follows from transformation noncontextuality.

Thanks to this result we get the following algorithmic formulation for Lemma~\ref{lem:transformationNC}: consider the vertices $p_{M_i}(m|\kappa)$ from Sec.~\ref{algsetup} that satisfy the additional constraint $p_{M_1}(m|\kappa)\in\{0,1\}$. To determine a set of achievable $p(m|M_i,P_k)$ we consider outcome-deterministic measurement noncontextual models given as
\begin{equation}
  p(m|M_i,P_*) = \sum_{\kappa} p_*(\kappa) p_{M_i}(m|\kappa).
\end{equation}
where the sum is over the vertices $\kappa$ satisfying the determinism constraint.
Of the $16$ vertices determined before, we find $12$ satisfy it and store them in 12 by 8 matrix \textsf{mncdetverticeswvN}.

We now project the 12 vertices down to the subspace that corresponds to the operational quantities we want to relate via noncontextuality: $p_-, p_F$. This subspace corresponds to the coordinates $x_1$ and $x_5$: $x_1$ is for the effect $[1|M_F]$ (hence related to $p_F$), and $x_5$ for $[X=-1,y=1|M]$ (hence related to $p_-$). This is done by restricting the 12 vertices to the coordinates $(x_1,x_5)$ and constructing their convex hull to yield the reduced polytope. This reduced polytope, named \textsf{mncreduceddetpolyN}, constructed in this subspace, $\mathbb{R}^2$ has 4 vertices. By trying several values of $p_d$ and $\tilde{p}$ we find they are of the form $(0, 0)$, $(0, p'_d)$, $(1, 0)$, $(1, \tilde{p})$ where $p'_d = \min\{p_d,\tilde{p}\}$ which will equal $p_d$ for typical parameters. The H-representation of the polytope is given by: $x_1,x_5\geq 0$, $x_1\leq 1$, and $x_1\frac{\tilde{p}-p'_d}{p'_d\tilde{p}}-\frac{x_5}{p'_d\tilde p}+\frac{1}{\tilde p} \geq 0$. The last inequality gives an operational constraint of $p_- \leq p_F \tilde p + (1-p_F)p'_d = \min\{p_F \tilde p + (1-p_F)p_d, \tilde p\}$, as expected from Lemma~\ref{lem:transformationNC}. 

\subsection{Analysis of the inequality in Lemma~\ref{lem:preparationNC}}

For this case no new tricks are required and so we very closely follow \cite{schmid17}. If there is a measurement noncontextual model then the observed statistics $p(m|M_i,P_k)$ can be written as
\begin{equation}
  p(m|M_i,P_k) = \sum_{\kappa} p_k(\kappa) p_{M_i}(m|\kappa),
\end{equation}
where we now sum over all 16 vertices $\{\kappa\}$.
Preparation noncontextuality applied to Eq.~\eqref{eq:algoprep} gives
\begin{equation}
\label{eq:algopreparation}
	q_\perp p_1(\kappa)+(1-q_\perp)p_2(\kappa)=q_0 p_3(\kappa)+(1-q_0)p_4(\kappa), \quad \forall \kappa.
\end{equation}

We thus we arrive at the following formulation for Theorem~\ref{th:preparationNC}. In order for a non-contextual model to satisfy the assumptions of Theorem~\ref{th:preparationNC} and reproduce the statistics $p(m|M_i,P_j)$, the following constraints must be satisfied
\begin{align}
	&\forall \kappa, j: p_j(\kappa)\geq 0,\\
	&\forall j: \sum_{\kappa}p_j(\kappa)=1,\\
	&\forall \kappa: q_\perp p_1(\kappa)+(1-q_\perp)p_2(\kappa)-(q_0 p_3(\kappa)+(1-q_0)p_4(\kappa))=0,\label{eq:kappaequiv}\\
	&\forall i,j,m: \sum_{\kappa}p_{M_i}(m|\kappa)p_j(\kappa)=p(m|M_i,P_j).
\end{align}
The problem can be solved by eliminating the variables $p_i(\kappa), i\in\{1,2,3,4\}, \kappa\in\{\kappa_1,\kappa_2,\dots,\kappa_{16}\}$. 
Since we don't care about all of the $p(m|M_i,P_j)$, for computational efficiency we first take the 16 vertices of the polytope \textsf{weakvaluespolysymbN} and cull all the coordinates from them except $x_1,x_5$. This is because we want to look at constraints from noncontextuality on the quantities $(p_F, C_{S}, p_-)$ which are a function of these two coordinates alone.

Using the resulting set of 16 vertices projected in the $(x_1,x_5)$ subspace, denoted \textsf{mncx1x5verticesN}, as an input to SageMath's \textsf{Polyhedron()}, we obtain a $2$-dimensional $5$-vertex polytope in $\mathbb{R}^2$, denoted \textsf{mncreducedpolyN}. We keep the vertices $\kappa'$ of this polytope in a 5 by 2 matrix \textsf{redvtxN}.

The problem is now to eliminate $p_j(\kappa'), j\in\{1,2,3,4\}, \kappa'\in\{\kappa'_1,\kappa'_2,\dots,\kappa'_5\}$. Using Eq.~\eqref{eq:kappaequiv} we can manually eliminate $p_1(\kappa') = \frac{1}{q_\perp}(q_0p_3(\kappa')+(1-q_0)p_4(\kappa')-(1-q_\perp)p_2(\kappa'))$ and arrive at the following constraints:
	\begin{eqnarray}
		\forall \kappa',\forall j\in\{2,3,4\}&:& p_j(\kappa')\geq 0,\\
		\forall \kappa'&:& q_0p_3(\kappa')+(1-q_0)p_4(\kappa')-(1-q_\perp)p_2(\kappa')\geq 0,\\
		\forall j\in\{2,3,4\}&:& \sum_{\kappa'}p_j(\kappa')=1,\\
		&&\sum_{i=1}^5 p_2(\kappa'_i)\kappa'_i(0)=p(1|M_2,P_2)\equiv p_F,\\
		&&\sum_{i=1}^5 p_2(\kappa'_i)\kappa'_i(1)=p(1|M_4,P_2)\equiv p_-,\\
		&&q_0\left(\sum_{i=1}^5p_3(\kappa'_i)(1-\kappa'_i(0))\right)+(1-q_0)\left(\sum_{i=1}^5p_4(\kappa'_i)\kappa'_i(0)\right)
		=C_{S}.
	\end{eqnarray}
(Here $\kappa'_i(a)$ denotes the $a$th entry of vertex $\kappa'_i=(\kappa'_i(0),\kappa'_i(1))$, where $i\in\{1,2,3,4,5\}$, $a\in\{1,2,3\}$).

We now carry out the remaining eliminations as follows. We construct the polytope of vectors $((p_2(\kappa'_i))_{i=1}^5,(p_3(\kappa'_i))_{i=1}^5,(p_4(\kappa'_i))_{i=1}^5, p_F,C_{S},p_-)$ in $\mathbb{R}^{18}$ subject to the above constraints. This is a 12-dimensional polytope in $\mathbb{R}^{18}$ with 45 vertices, denoted \textsf{robustawvpolyN}.

We project the vertices down to just the coordinates $(p_F,C_{S},p_-)$ and construct a polytope with these as an input to \textsf{Polyhedron()}. This results in the polytope named \textsf{redawvpolyN}, a 3-dimensional polytope in $\mathbb{R}^3$ with 10 vertices.

For a representative case of $\left(q_*, q_\perp, p_d, \tilde{p}\right)=\left(\frac12,\frac12,\frac14,\frac12\right)$ the facets of this polytope include our noncontextuality inequality Eq.~\eqref{eq:inequalitypnc}: $p_F-4C_{S}-4p_-+5\geq 0$ or
\begin{align}
  p_- \leq \frac{p_F}{2}+\frac{1-p_F}{4}+\left(1-\frac{1}{2}\right)\frac{1-C_{S}}{1/2}
  =\tilde p p_F+p_d(1-p_F)+(1-\tilde p)\frac{1-C_{S}}{p_*}.
\end{align}

The overall tradeoff between $(p_F,C_{S},p_-)$ for this case is depicted in Figure \ref{tradeoffplot}. 
\begin{figure}
  \includegraphics[scale=0.7]{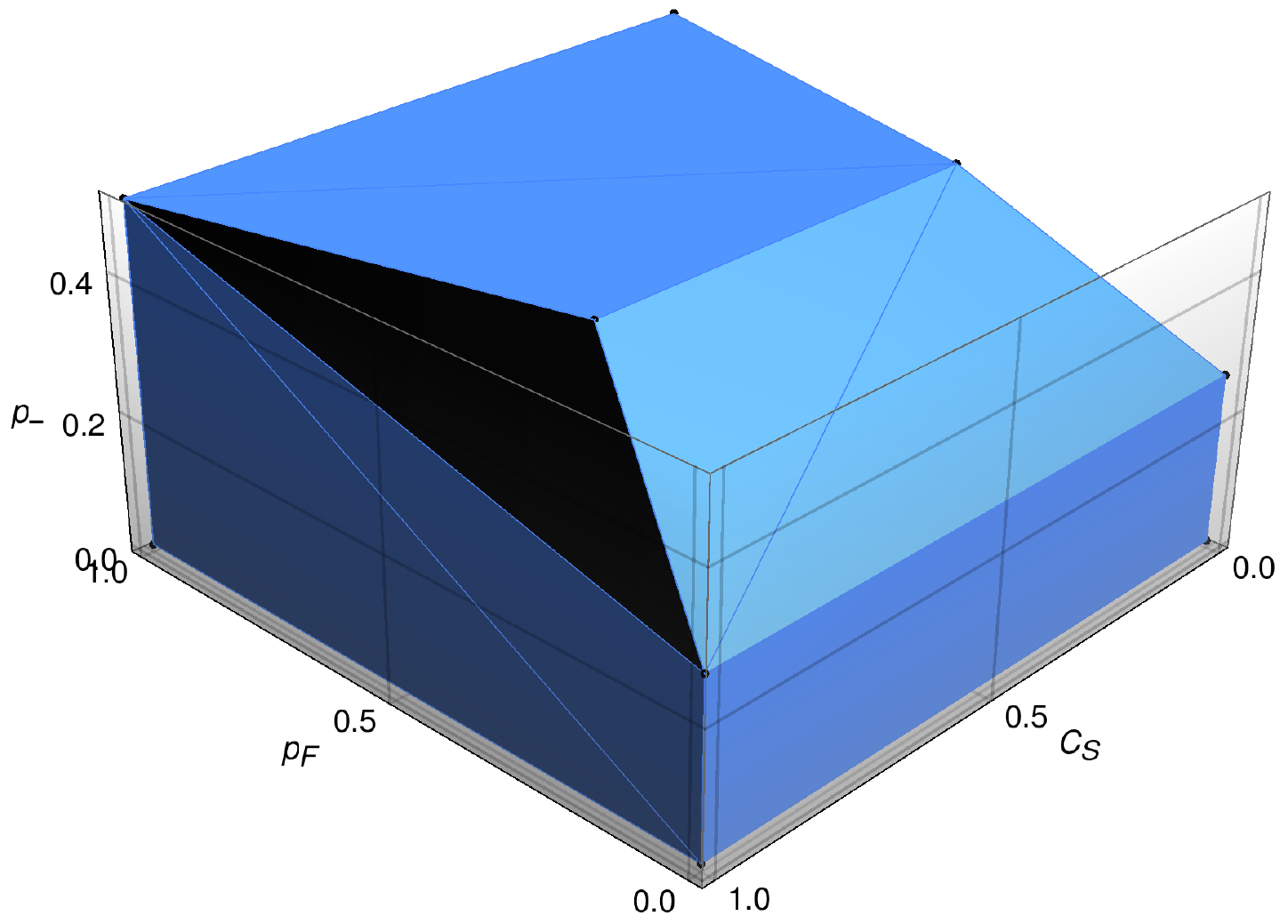}
  \caption{The noncontextuality tradeoff between $p_-,p_F$ and $C_{S}$ for $p_d=1/4,\tilde p=1/2, q_0=q_*=1/2$. The facet corresponding to \eqref{eq:th2inequality} is shown in black.}
	\label{tradeoffplot}
\end{figure}  

We also tried many other values of $\left(q_*, q_\perp, p_d, \tilde{p}\right)$. Eq.~\eqref{eq:inequalitypnc} always appeared as a facet, except when $p_d > \tilde p$.
\end{appendix}

\end{document}